\documentclass{article}

\usepackage{arxiv}

\usepackage[utf8]{inputenc} % allow utf-8 input
\usepackage[T1]{fontenc}    % use 8-bit T1 fonts
\usepackage{hyperref}       % hyperlinks
\usepackage{url}            % simple URL typesetting
\usepackage{booktabs}       % professional-quality tables
\usepackage{amsfonts}       % blackboard math symbols
\usepackage{nicefrac}       % compact symbols for 1/2, etc.
\usepackage{microtype}      % microtypography
\usepackage{lipsum}		% Can be removed after putting your text content
\usepackage{graphicx}

\usepackage{dcolumn}% Align table columns on decimal point
\usepackage{bm}% bold math
\usepackage{amsmath}
\usepackage{amssymb}
\usepackage{amscd}
\usepackage{amsthm}
\usepackage{amsfonts}
\usepackage{graphicx}
\usepackage{multirow}

\newcommand{\beq}{\begin{equation}}
\newcommand{\eeq}{\end{equation}}
\newcommand{\bea}{\begin{eqnarray}}
\newcommand{\eea}{\end{eqnarray}}
\newcommand{\bc}{\begin{cases}}
\newcommand{\ec}{\end{cases}}

\newtheorem{theorem}{Theorem}

\newtheorem{definition}[theorem]{Definition}
\newtheorem{example}[theorem]{Example}

\newtheorem{remark}[theorem]{Remark}

\date{\today}

\author{
Jos\'{e} M. Amig\'{o} \\
Centro de Investigaci\'{o}n Operativa, \\
Universidad Miguel Hern\'{a}ndez, \\
Elche, 03202 Alicante, Spain \\
\texttt{jm.amigo@umh.es}
\And
Roberto Dale \\
Centro de Investigaci\'{o}n Operativa, \\
Universidad Miguel Hern\'{a}ndez, \\
Elche, 03202 Alicante, Spain \\
\texttt{rdale@umh.es}
\And
Piergiulio Tempesta \\
Departamento de F\'{\i}sica Te\'{o}rica, \\
Facultad de Ciencias F\'{\i}sicas, \\
Universidad Complutense de Madrid, \\
28040 Madrid, Spain \\
Instituto de Ciencias Matem\'aticas, \\
28049 Madrid, Spain \\
\texttt{p.tempesta@fis.ucm.es, piergiulio.tempesta@icmat.es}
}

\begin{document}

\title{Complexity-based permutation entropies: from deterministic time series to white noise}
\maketitle

\begin{abstract}
This is a paper in the intersection of time series analysis and complexity
theory that presents new results on permutation complexity in general and
permutation entropy in particular. In this context, permutation complexity
refers to the characterization of time series by means of ordinal patterns
(permutations), entropic measures, decay rates of missing ordinal patterns,
and more. Since the inception of this \textquotedblleft
ordinal\textquotedblright\ methodology, its practical application to any
type of scalar time series and real-valued processes have proven to be
simple and useful. However, the theoretical aspects have remained limited to
noiseless deterministic series and dynamical systems, the main obstacle
being the super-exponential growth of allowed permutations with length when
randomness (also in form of observational noise) is present in the data. To
overcome this difficulty, we take a new approach through complexity classes,
which are precisely defined by the growth of allowed permutations with
length, regardless of the deterministic or noisy nature of the data. We
consider three major classes: exponential, sub-factorial and factorial. The
next step is to adapt the concept of Z-entropy to each of those classes,
which we call permutation entropy because it coincides with the conventional
permutation entropy on the exponential class. Z-entropies are a family of
group entropies, each of them extensive on a given complexity class. The
result is a unified approach to the ordinal analysis of deterministic and
random processes, from dynamical systems to white noise, with new concepts
and tools. Numerical simulations show that permutation entropy discriminates
time series from all complexity classes.
\end{abstract}

% keywords can be removed

\tableofcontents

\section{Introduction}

\label{sec:1}

Complexity in symbolic times series, symbols being taken from a finite
alphabet $\mathcal{A}$, has to do with the number of different sequences
(strings, words, blocks,...) of a given length $n$ and how this number
increases with $n$. The perhaps simplest approach consists in counting the
number of such sequences. In this case, the complexity of periodic sequences
is a bounded function of $n$ \cite{Morse1940}, while the complexity of
arbitrary sequences grows as $\left\vert \mathcal{A}\right\vert ^{n}$ ($%
\left\vert \cdot \right\vert $ denotes cardinality). Take the logarithmic
growth rate, namely $\log \left\vert \mathcal{A}\right\vert $, to obtain the
Shannon entropy of a memoryless process that outputs the symbols of $%
\mathcal{A}$ with equal probabilities. The positivity of the entropy
differentiates then exponential from sub-exponential growth. Other
approaches to the concept of complexity of sequences and the processes
producing them have been proposed in different fields. Thus, in information
theory complexity is usually related to compression \cite{Lempel1976,Ziv1978}%
. Here one counts the number of new words arising as one parses the whole
message (ideally, a one-sided infinite binary sequence). In dynamical
systems and symbolic dynamics, the main tool is the dynamic entropy, both in
its metric and topological versions \cite{Amigo2015}. In computer science,
algorithmic (or Kolmogorov) complexity refers to the shortest computer code
that generates the sequence at hand, while computational problems are
grouped into (polynomial, exponential,...) complexity classes according to
how the amount of resources (time, memory,...) needed to solve them using a
computation model (Turing machine, probabilistic Turing machine, quantum
computer,...) depends on the \textquotedblleft size\textquotedblright\ of
the input\ (usually, the number of bits) \cite{Li2008}. In number theory and
cryptography there are also several proposals, some of them going deep into
the concepts of randomness, compressibility and typicality \cite%
{Volchan2002,Shen2017,Downey2019}.

This paper deals with the concept of permutation complexity of real-valued
time series introduced in \cite{Amigo2010B,Amigo2010,Monetti2013}, so our
symbols will be ordinal patterns or permutations of length $L$ \cite%
{Bandt2002}. As we will see more precisely in the following sections, the
count of permutations grows exponentially with $L$ in the case of
(noiseless) deterministic signals, while it grows super-exponentially for
noisy deterministic and random signals, sometimes called \textit{noisy
signals} hereafter for brevity. This different growth behavior of the
ordinal patterns and, therefore, of the permutation complexity makes
possible to distinguish deterministic signals from noisy signals but, at the
same time, it poses a challenge for a unified quantification of permutation
complexity for a simple reason: the usual tools for measuring complexity
(say, Shannon and Kolmogorov-Sinai entropies) are designed for exponential
growths of the symbols they are defined upon, thus diverging when applied to
super-exponential growths. This occurs, in particular, with permutation
entropy, which is the Shannon entropy of a time series in its \textit{%
ordinal representation}, i.e., its symbolic representation via ordinal
patterns.

As a result, the tools of permutation complexity are applied to time series
analysis in different ways. In theoretical applications, where the time
series are deterministic and may be assumed to be infinitely long, one uses
entropic measures such as metric and topological permutation entropy, or the
like. In practical applications, where the time series are noisy and finite,
one typically resorts to permutations entropies of finite order (Section \ref%
{sec2.1}), causality-complexity planes \cite{Rosso2007,Zunino2012}, the
decay rate of the missing patterns \cite{Amigo2007,Carpi2010}, and ordinal
networks \cite{Pessa2019}, to mention some typical techniques. This being
the case, the objective of the present paper is to propose an integrating
and overarching approach as follows.

The main character of this new approach is the logarithmic growth of allowed
(or visible) ordinal patterns with increasing length. Depending on that
growth, processes are collected in the exponential, sub-factorial and
factorial complexity classes, whether they are deterministic or random. This
procedure was inspired by similar ideas in complexity theory, where systems
are usually classified according to the state growth rates of the states
with the number of constituents $N$. For each of those classes there is a
particular \textit{group entropy}, called $Z$-entropy, that is \textit{%
extensive} for the systems in the class, meaning that it is finite over the
uniform probability distributions in the limit $N\rightarrow \infty $ \cite%
{TJ2020,PT2020}. In our context, system translates into process, the
extensive parameter $N$ into the length $L$ of the ordinal patterns, $Z$%
-entropy into (generalized) permutation entropy, and extensivity into the
convergence of the corresponding topological permutation entropy rate.
Nevertheless, the introduction of the key concepts will be self-contained
and will concentrate on time series and processes, so that the reader can
understand their rationale and properties without further reference. The
result is a characterization of time series in the ordinal representation
that focuses on complexity rather than data generation. This way we extend
the realm of the standard permutation entropy from deterministic processes
(dynamical systems) to random processes, thus filling a conceptual gap in
permutation complexity. A first step in this direction was taken in \cite%
{Amigo2021}, where we used the $Z$-entropy of the factorial complexity class
to define a generalized permutation entropy for noisy signals without
forbidden patterns, i.e., noisy dynamics and random processes such that all
ordinal patterns of any length are allowed; these are the kind of signals
encountered in practice. Our approach here is more general and comprehensive.

Regarding the notion of group entropy mentioned above, it was introduced in 
\cite{PT2011PRE} and discussed, e.g.,\negthinspace\ in \cite%
{PT2016PRA,JT2018ENT,JPPT2018JPA,RRT2019PRA,TJ2020,PT2020}. Essentially, a
group entropy is a functional defined on a probability space which satisfies
several important properties, such as the first three Shannon-Khinchin
axioms (Section \ref{sec4}) and a so-called composability axiom: the entropy
of a system compound by two statistically independent systems is expressed
by a formal group law \cite{Amigo2021}. By construction, group entropies
have a direct interpretation as information measures \cite{RRT2019PRA,TJ2020}%
. In particular, they can be used to define divergences and Riemannian
structures over statistical manifolds.

The ordinal approach, where the information contained in the ordinal
patterns is exploited via probability distributions, entropies, etc., is
quite popular in time series analysis for a number of reasons, including its%
\textbf{\ }computational simplicity and speed. Applications to biomedicine
where among the first and include epilepsy \cite{Keller2004}, cardiopathies 
\cite{Parlitz2012}, heart rate variability \cite{Graff2013}, and more \cite%
{Amigo2015B}. Further applications include dynamical change detection \cite%
{Cao2004}, signal characterization \cite%
{Carpi2010,Olivares2019A,Olivares2019B}, and image processing \cite%
{Zunino2016,Chagas2021}. Currently, ordinal techniques, alone or
complemented by other methods, are being applied in plenty of fields, e.g.,
chaotic dynamics, earth science, computational neuroscience, and
econophysics; see \cite{Amigo2013} for examples, and \cite{Zanin2012} for a
recent survey.

The rest of this paper is organized as follows. Section \ref{sec2} contains
the mathematical setting for the subsequent discussion, in particular,
metric and topological permutation entropies as well as the concepts of
allowed and forbidden patterns. In doing so, we cover the full range of
discrete-time, real-valued time series envisaged in this paper, namely:
noiseless deterministic, noisy deterministic, and random signals. This
section is partially based on our paper \cite{Amigo2021}. Section \ref{sec3}
is devoted to the permutation complexity function and classes. Here we
introduce the exponential, sub-factorial and factorial permutation
complexity classes that are further analyzed in the subsequent sections. In
Section \ref{sec4} we briefly review the general concept of entropy (based
on the Shannon-Khinchin axioms), before extending permutation entropy from
the exponential class to the factorial and sub-factorial classes. Numerical
simulations is the subject of Section \ref{sec5}. In this section, the
discriminatory power of the permutation entropy (Section \ref{sec51}) and
the permutation complexity function (Section \ref{sec52}) is put to the test
with a battery of seven noisy processes from the factorial class. In Section %
\ref{sec53} we study numerically and analytically a toy model for
sub-factorial processes. The conclusions are summarized in Section \ref{sec6}%

\section{Permutation complexity}

\label{sec2}

Real-valued time series typically result from sampling analog signals or
observing dynamical flows at discrete times. A further step in the analysis
of such series can be the discretization of the data, a procedure that is
usually called symbolic representation. The information provided by a
symbolic representation may be sufficient for the intended application while
simplifying the mathematical tools needed for the analysis. In this regard,
ordinal patterns \cite{Bandt2002} are becoming increasingly popular to
represent symbolically real-valued time series. Some reasons for this is
their mathematically sound relation to Kolmogorov-Sinai entropy via
permutation entropy \cite{Bandt2002B,Keller2010,Amigo2012,Keller2019} and
their ease of computation. Ordinal patterns and permutation entropies are
the main ingredients of permutation complexity.

\subsection{Ordinal representations and permutation entropy}

\label{sec2.1}

Given a (finite or infinite) time series $(x_{t})_{t\geq
0}=x_{0},x_{1},\ldots ,x_{t},...$, where $t=0,1,...,N\leq \infty $ is
discrete time and $x_{t}\in \mathbb{R}$, its symbolic representation by 
\textit{ordinal patterns of length} $L\geq 2$ is $\mathbf{r}_{0},\mathbf{r}%
_{1},\ldots ,\mathbf{r}_{t},\ldots $, where $\mathbf{r}_{t}$ is the rank
vector of the string $x_{t}^{L}:=x_{t},x_{t+1},\ldots ,x_{t+L-1}$ ($t\leq
N-L+1$), i.e., $\mathbf{r}_{t}=(\rho _{0},\rho _{1},\ldots ,\rho _{L-1})$
where $\{\rho _{0},\rho _{1},\ldots ,\rho _{L-1}\}\in \{0,1,\ldots ,L-1\}$
are such that%
\begin{equation}
x_{t+\rho _{0}}<x_{t+\rho _{1}}<\ldots <x_{t+\rho _{L-1}}  \label{ord patt}
\end{equation}%
(other rules can also be found in the literature). In case of two or more
ties, one can adopt some convention, e.g., the earlier entry is smaller.
Sometimes we say that $x_{t}^{L}$ defines the \textit{ordinal }$L$\textit{%
-pattern} $\mathbf{r}_{t}$ or that it is of \textit{type} $\mathbf{r}_{t}$.
Ordinal $L$-patterns can be identified with permutations of $\{0,1,\ldots
,L-1\}$, i.e., with elements of the symmetric group of degree $L$, $\mathcal{%
S}_{L}$; the cardinality of $\mathcal{S}_{L}$, $\left\vert \mathcal{S}%
_{L}\right\vert $, is $L!$. Symbolic representations of time series by means
of ordinal patterns are called \textit{ordinal representations}. The
algebraic structure of $\mathcal{S}_{L}$ was exploited in \cite{Monetti2013}%
, which led to the more general concept of algebraic representations.

Furthermore, the time series $(x_{t})_{t\geq 0}$ is assumed to be output by
a discrete-time deterministic or random process $\mathbf{X}$ taking values
on an interval $I\subset \mathbb{R}$. By deterministic process we mean a one
dimensional dynamical system $(I,\mathcal{B},\mu ,f)$, where $I$ (the state
space) is a bounded interval of $\mathbb{R}$, $\mathcal{B}$ is the Borel $%
\sigma $-algebra of $I$, $\mu $ is a measure over the measurable space $(I,%
\mathcal{B})$ such that $\mu (I)=1$ (i.e., $(I,\mathcal{B},\mu )$ is a
probability space) and, for the time being, $f:I\rightarrow I$ is any $\mu $%
-invariant map (i.e., $\mu (f^{-1}(B))=\mu (B)$ for all $B\in \mathcal{B}$);
alternatively, we say that $\mu $ is $f$-invariant. In this case, the output 
$(x_{t})_{t\geq 0}$ of $\mathbf{X}$ is the orbit of $x_{0}$, i.e., $%
(x_{t})_{t\geq 0}=(f^{t}(x_{0}))_{t\geq 0}$, where $f^{0}(x_{0})=x_{0}\in I$
and $f^{t}(x_{0})=f(f^{t-1}(x_{0}))$. An ordinal representation of the
orbits of $f$ by ordinal $L$-patterns partitions the state space $I$ into
the $L!$ bins%
\begin{equation}
P_{\mathbf{r}}=\{x\in I:\text{ }(x,\,f(x),...,\,f^{L-1}(x))\text{ is of type 
}\mathbf{r}\in \mathcal{S}_{L}\}.  \label{P_r}
\end{equation}%
Therefore, the probability $p(\mathbf{r})$ of the ordinal pattern $\mathbf{r}%
\in \mathcal{S}_{L}$ to occur in an output of the deterministic process $%
\mathbf{X}$ generated by the map $f$ is 
\begin{equation}
p(\mathbf{r})=\mu (P_{\mathbf{r}}).  \label{p(r)}
\end{equation}%
Note that, although the outputs $(x_{t})_{t\geq 0}$ are deterministic
(\textquotedblleft sharp\textquotedblright\ orbits), their ordinal
representations $(\mathbf{r}_{t})_{t\geq 0}$ are random sequences
(\textquotedblleft pixelated\textquotedblright\ orbits), as occurs with any
symbolic dynamics of a map with respect to a partition of its state space 
\cite{Amigo2010B}.

The \textit{metric permutation entropy} (\textit{rate}) of the process $%
\mathbf{X}$ is defined as%
\begin{equation}
h^{\ast }(\mathbf{X})=\underset{L\rightarrow \infty }{\lim \sup }\frac{1}{L}%
H^{\ast }(X_{0}^{L}),  \label{h*_mu}
\end{equation}%
where $X_{0}^{L}=X_{0},X_{1},...,X_{L-1}$ and 
\begin{equation}
H^{\ast }(X_{0}^{L})=-\sum_{\mathbf{r}\in \mathcal{S}_{L}}p(\mathbf{r})\ln p(%
\mathbf{r})  \label{h*_mu,L}
\end{equation}%
is the \textit{metric permutation entropy of\ }$\mathbf{X}$ \textit{of order}
$L$.

In other words, $H^{\ast }(X_{0}^{L})$ is the Shannon entropy of the
probability distribution $\{p(\mathbf{r}):\mathbf{r}\in \mathcal{S}_{L}\}$.
If $\mathbf{X}$ is a deterministic process (and $\mu $ is known), then $p(%
\mathbf{r})$ is given as in Equations (\ref{P_r})-(\ref{p(r)}). If $\mathbf{X%
}$ is a random process, the probabilities $p(\mathbf{r})$ can only
exceptionally be derived from the probability distributions of $\mathbf{X}$ 
\cite{Bandt2007} so, in general, they have to be estimated, e.g. by relative
frequencies:%
\begin{equation}
\hat{p}(\mathbf{r})=\frac{\left\vert \{x_{t}^{L}\text{ of type }\mathbf{r}%
\in \mathcal{S}_{L}:0\leq t\leq N-L+1\}\right\vert }{N-L+2}.  \label{nu(r)}
\end{equation}%
In the theoretical case of an infinite time series, take the limit $%
N\rightarrow \infty $ in (\ref{nu(r)}). In nonlinear time series analysis,
the ergodic invariant measure defined by $\mu (P_{\mathbf{r}})=\hat{p}(%
\mathbf{r})$ is called the physical or natural measure because it is the
only relevant measure for physical systems and numerical simulations \cite%
{Eckmann1985}. More about this in Section \ref{sec51}.

\begin{remark}
\label{RemarkStationarity}The limit $\lim_{N\rightarrow \infty }\hat{p}(%
\mathbf{r})$ exists with probability $1$ when the underlying stochastic
process fulfills the following weak stationarity condition: for $k\leq L-1,$
the probability for $x_{t}<x_{t+k}$ should not depend on $t$ \cite{Bandt2002}%
. This is the case for stationary processes but also for non-stationary
processes with stationary increments such as the fractional Brownian motion 
\cite{Mandelbrot1968} and its increments, that is, the fractional Gaussian
noise. We will use these random processes, which have long range
dependencies, in the numerical simulations.
\end{remark}

Let $\mathbf{X}$ be a deterministic or random process that takes values on
an interval $I\subset \mathbb{R}$. We say that an ordinal pattern $\mathbf{r}%
\in \mathcal{S}_{L}$ is \textit{allowed} for $\mathbf{X}$ if the probability
that a string $x_{t}^{L}$ of type $\mathbf{r}$ is output by $\mathbf{X}$ is
positive. That is, an $L$-pattern $\mathbf{r}$\ is allowed if there are
strings $x_{t},...,x_{t+L-1}$\ in some outputs or orbits of $\mathbf{X}$\
such that the type of those strings is $\mathbf{r}$. Otherwise, the ordinal $%
L$-pattern $\mathbf{r}$ is \textit{forbidden} for $\mathbf{X}$. For example,
the ordinal $3$-pattern $\mathbf{r}=(2,1,0)$ is forbidden for the logistic
map $f(x)=4x(1-x)$, $0\leq x\leq 1$, because there is no string $%
x_{t},x_{t+1},x_{t+2}$ in any orbit of $f$ such that $x_{t+2}<x_{t+1}<x_{t}$%
; all other $3$-patterns $\mathbf{r}=(\rho _{0},\rho _{1},\rho _{2})$, where 
$\rho _{0},\rho _{1},\rho _{2}\in \{0,1,2\}$ and $\mathbf{r}\neq (2,1,0)$,
are allowed for the logistic map, that is, $x_{t+\rho _{0}}<x_{t+\rho
_{1}}<x_{t+\rho _{2}}$ for $x_{t}$ in a suitable subinterval of $[0,1]$, see 
\cite{Amigo2010B}. Since we do not consider patterns other than ordinal
patterns in this paper, we speak of allowed and forbidden patterns for
brevity.

If $\mathcal{A}_{L}(\mathbf{X})$ denotes the number of allowed patterns of
length $L$ for $\mathbf{X}$, the \textit{topological permutation entropy} (%
\textit{rate}) of the process $\mathbf{X}$ is then defined as%
\begin{equation}
h_{0}^{\ast }(\mathbf{X})=\underset{L\rightarrow \infty }{\lim \sup }\frac{1%
}{L}H_{0}^{\ast }(X_{0}^{L}),  \label{h*_0}
\end{equation}%
where%
\begin{equation}
\;H_{0}^{\ast }(X_{0}^{L})=\ln \mathcal{A}_{L}(\mathbf{X})  \label{h*_0,L}
\end{equation}%
is the \textit{topological permutation entropy of\ }$\mathbf{X}$ \textit{of
order} $L$. Moreover,%
\begin{equation}
H^{\ast }(X_{0}^{L})\leq H_{0}^{\ast }(X_{0}^{L})\leq \ln L!,  \label{bounds}
\end{equation}%
where $H^{\ast }(X_{0}^{L})=H_{0}^{\ast }(X_{0}^{L})$ for flat probability
distributions of the allowed $L$-patterns, and $H_{0}^{\ast }(X_{0}^{L})=\ln
L!$ if all $L$-patterns are allowed.

\subsection{Allowed pattern growths for deterministic and random processes}

\label{sec2.2}

The map $f:I\rightarrow I$ is called \textit{piecewise monotone} if there is
a finite partition of $I$ such that $f$ is continuous and strictly monotone
on each subinterval of the partition. If the graph of $f$ has $n$ humps,
then $f$ is called unimodal $(n=1)$ or multimodal ($n>1$). Most
one-dimensional maps encountered in practice are piecewise monotone, so this
condition does not imply any strong restriction for practical purposes. Let $%
h_{0}(f)$ denote the topological entropy of $f$, and $h(f)$ its metric (or
Kolmogorov-Sinai) entropy \cite{Walter2000}. The following theorem holds 
\cite{Bandt2002B}.

\begin{theorem}
\label{Thm1}If $f$ is piecewise monotone, then (a) $h^{\ast }(f)=h(f)$, and
(b) $h_{0}^{\ast }(f)=h_{0}(f)$.
\end{theorem}

Theorem \ref{Thm1}(a) was generalized to countably piecewise monotone maps
in \cite{Keller2019}. Generalizations to higher dimensional intervals can be
found in \cite{AmigoKeller2013}.

From Theorem \ref{Thm1}(b) and Equations (\ref{h*_0})-(\ref{h*_0,L}) it
follows that 
\begin{equation}
\ln \mathcal{A}_{L}(\mathbf{X})=\ln \left\vert \{\text{allowed }L\text{%
-patterns for deterministic }\mathbf{X}\}\right\vert \sim h_{0}(f)L,
\label{allowed pat f}
\end{equation}%
where $f$ is the map generating the outputs of $\mathbf{X}$, the symbol $%
\sim $ stands for \textquotedblleft asymptotically when $L\rightarrow \infty 
$\textquotedblright\ (i.e., $\lim_{L\rightarrow \infty }\ln \mathcal{A}_{L}(%
\mathbf{X})/(h_{0}(f)L)=1$) and, for the sake of this paper, we assume $%
h_{0}(f)>0$ throughout. Therefore, the number of allowed $L$-patterns for a
piecewise monotone map $f$ grows exponentially with $L$. To be more precise,
according to the proof of Proposition 1(b) in \cite{Bandt2002B}, $\mathcal{A}%
_{L}(\mathbf{X})\sim \exp [h_{0}(f)L+\ln L+const]$.

\begin{remark}
\label{Remarkiff}More generally, it is easy to show that $\ln \phi (L)\sim
g(L)$ if and only if $\phi (L)=\exp [g(L)+o(g(L)]$, where $o(g(L))$ denotes
a function such that $o(g(L))/g(L)\rightarrow 0$ when $L\rightarrow \infty $%
. Exponential growth of $\phi (L)$ with respect to $L$ corresponds to $g(L)$
linear, as in Equation (\ref{allowed pat f}). In that particular case: $%
\mathbf{X}=f$, $\phi (L)=\mathcal{A}_{L}(\mathbf{X})$, $g(L)=h_{0}(f)L$ and $%
o(g(L))=\ln L+const=o(L)$.
\end{remark}

Since, on the other hand, the number of possible $L$-patterns is $L!$ and 
\begin{equation}
\ln L!\sim L\ln L  \label{Stirling}
\end{equation}%
by Stirling's formula $\ln L!\simeq L(\ln L-1)+\tfrac{1}{2}\ln (2\pi L)$, we
conclude from Equation (\ref{allowed pat f}) that deterministic processes
necessarily have forbidden $L$-patterns for $L$ large enough and, in fact,
the number of forbidden $L$-patterns grows super-exponentially with $L$. By 
\textit{deterministic process} we mean here and hereafter the dynamics
generated by (the iteration of) a piecewise monotone map so that Theorem \ref%
{Thm1} is applicable and Equation (\ref{allowed pat f}) holds with $%
h_{0}(f)>0$. Sometimes we write $\mathbf{X}=f$ in this case.

As mentioned before, $\mathbf{r}=(2,1,0)$ is the only forbidden $3$-pattern
for the logistic parabola, while all $3$-patterns are allowed for the shift
map $x\mapsto 2x~\mathrm{{mod}~1}$ \cite{Amigo2010B}. The respective number
of forbidden $4$-patterns is 12 and 6 \cite{Amigo2010B}. Each forbidden
pattern of length $L_{0}$ in a deterministic dynamic is the seed of an
infinitely long trail of \textquotedblleft outgrowth forbidden
patterns\textquotedblright\ of lengths $L>L_{0}$ whose structure can be
found in \cite{Amigo2010B}. Let us mention in passing that forbidden
patterns there exist also in higher dimensional dynamics (at least) for
expansive maps, ordinal patterns being defined lexicographically \cite%
{Amigo2008B}. Therefore, projections of higher dimensional dynamics are
expected to have forbidden patterns and exponential growths of allowed
patterns as well.

At the other extreme are random processes without forbidden patterns, that
is, processes for which all ordinal patterns of any length are allowed and,
hence, their growth is factorial: $\mathcal{A}_{L}(\mathbf{X})=\left\vert 
\mathcal{S}_{L}\right\vert =L!$. A trivial example of a random process
without forbidden patterns is white noise.

Also noisy deterministic time series may not have forbidden patterns (for
sufficiently long series). Indeed, when the dynamics takes place on a
nontrivial attractor so that the orbits are dense, then the observational
(white) noise will \textquotedblleft destroy\textquotedblright\ all
forbidden patterns in the long run, no matter how small the noise. For this
reason, we sometimes call noisy deterministic processes and other random
processes without forbidden patterns just \textit{forbidden-pattern-free }%
(FPF) \textit{processes} or signals. Unlike Equation (\ref{allowed pat f})
for deterministic processes, for FPF processes we have 
\begin{equation}
\ln \mathcal{A}_{L}(\mathbf{X})=\ln \left\vert \{\text{allowed }L\text{%
-patterns for FPF }\mathbf{X}\}\right\vert =\ln L!\sim L\ln L,
\label{allowed pat X}
\end{equation}%
where we used the asymptotic equivalence (\ref{Stirling}).

To complete the picture, let us point out that random processes can have
forbidden patterns too. A conceptually simple (though impractical) way of
constructing such a process is to repeatedly draw $x_{t}$ until the type of
the block $x_{0},x_{1},...,x_{t}$ is allowed, for $t=1,2,...$ By controlling
the number of allowed $L$-patterns, this constrained random process outputs
time series with any feasible growth of allowed $L$-patterns, in particular,
an exponential one (as in the deterministic case). A more realistic example
of a random process with a sub-factorial growth of allowed pattern is the
following.

\begin{example}
\label{ExampleSubfact}\emph{(Not-so-noisy measurement of a periodic signal)}%
. Suppose that a periodic time series $(y_{t})_{t\geq
0}=(f^{t}(y_{0}))_{t\geq 0}$ of prime period $\mathfrak{p}\geq 2$ is
observed; so%
\begin{equation*}
y_{t}=f^{k}(y_{0})=y_{k}
\end{equation*}%
for every $t=k~\mathrm{{mod}~}\mathfrak{p}$, $k=0,1,...,\mathfrak{p}-1$,
where $y_{0}<y_{1}<...<y_{\mathfrak{p}-1}$ for simplicity. Furthermore,
suppose that the points $y_{k}$ are measured with a device whose precision
is value dependent, so that only the measurement of, say, $y_{\mathfrak{p}%
-1} $ is noiseless and, otherwise, the uncertainty intervals of $%
y_{0},...,y_{\mathfrak{p}-2}$ do not overlap. To model this situation, let $%
\delta >0$ be the minimum separation between the points of the periodic
cycle $(y_{0},y_{1},...,y_{\mathfrak{p}-1})$ and add white noise to $y_{t}$
with amplitude less than $\delta /2$, except when $t=\mathfrak{p}-1~\mathrm{{%
mod}~}\mathfrak{p}$. That is, the noisy observations are $x_{t}=y_{t}+\zeta
_{t}$ for $t=0,1,...,\mathfrak{p}-2~\mathrm{{mod}~}\mathfrak{p}$, where $%
\zeta _{t} $ are independent and uniformly distributed random variables in $%
(-\delta /2,\delta /2)$, and $x_{t}=y_{t}$ for $t=\mathfrak{p}-1~\mathrm{{mod%
}~}\mathfrak{p}$, so that $x_{\nu \mathfrak{p}}<x_{\nu \mathfrak{p}%
+1}<...<x_{(\nu +1)\mathfrak{p}-1}=y_{\mathfrak{p}-1}$ for all $\nu \in 
\mathbb{N}$. Choose $L=\nu \mathfrak{p}$ for simplicity. Then the number of
allowed $L$-patterns is given by%
\begin{equation}
\mathcal{A}_{L}(\mathbf{X}_{\mathfrak{p}})=\mathfrak{p}\left( \nu !\right) ^{%
\mathfrak{p}-1}=\mathfrak{p}[(L/\mathfrak{p})!]^{\mathfrak{p}-1},
\label{formula}
\end{equation}%
where $\mathbf{X}_{\mathfrak{p}}$ is the noisy process that outputs the time
series $(x_{t})_{t\geq 0}$. The factor $\mathfrak{p}$ in Equation (\ref%
{formula}) comes from the $\mathfrak{p}$ different values of the time index $%
t$ modulus $\mathfrak{p}$. For each such $t$ (say, $t=0,1,..,\mathfrak{p}-1$%
), the window $x_{t}^{L}=x_{t},x_{t+1},...,x_{t+L-1}$ splits in $\mathfrak{p}
$ disjoint groups of $\nu $ points each as follows:%
\begin{equation}
\{x_{t+j}:t+j=0~\mathrm{{mod}~}\mathfrak{p}\}<\{x_{t+j}:t+j=1~\mathrm{{mod}~}%
\mathfrak{p}\}<....<\{x_{t+j}:t+j=\mathfrak{p}-1~\mathrm{{mod}~}\mathfrak{p}%
\}\mathrm{,}  \label{splitting}
\end{equation}%
\noindent where $0\leq j\leq L-1$ and $x_{t+j}=y_{\mathfrak{p}-1}$ for all $%
t+j=\mathfrak{p}-1~\mathrm{{mod}~}\mathfrak{p}$ (last group). As a result,
the time indices of the $\nu $ noisy points $x_{t+j}$ in each of the first $%
\mathfrak{p}-1$ groups of the splitting (\ref{splitting}) can be ordered in
any of the $\nu !$ permutations ($\nu $-patterns) possible, while the time
indices of the noiseless points $x_{t+j}=y_{\mathfrak{p}}$ in the last group
leads to only one $\nu $-pattern, namely:\ the permutation consisting of the
corresponding time indices in increasing order (according to the convention
for repeated values). This explains the second factor $\left( \nu !\right) ^{%
\mathfrak{p}-1}$ in Equation (\ref{formula}). Therefore, as $\nu =L/%
\mathfrak{p}$ increases, 
\begin{equation}
\ln \mathcal{A}_{L}(\mathbf{X}_{\mathfrak{p}})=\ln \mathfrak{p}+(\mathfrak{p}%
-1)\ln [(L/\mathfrak{p})!]\sim (\mathfrak{p}-1)\frac{L}{\mathfrak{p}}\ln 
\frac{L}{\mathfrak{p}}=\frac{\mathfrak{p}-1}{\mathfrak{p}}L(\ln L-\ln 
\mathfrak{p})\sim cL\ln L  \label{Example}
\end{equation}%
where $c=(\mathfrak{p}-1)/\mathfrak{p}<1$ and $L=\nu \mathfrak{p}$.
Obviously, if the number of \textquotedblleft noiseless\textquotedblright\
measurements of the periodic cycle $(y_{0},y_{1}...,y_{\mathfrak{p}-1})$ is
generalized to $m$, then, $c=(\mathfrak{p}-m)/\mathfrak{p}$.
\end{example}

The noisy process presented in Example \ref{ExampleSubfact} will be
discussed with greater detail in Section \ref{sec53}. In particular, the
asymptotic growth of $\ln \mathcal{A}_{L}(\mathbf{X}_{\mathfrak{p}})$
depends on $L$ modulus$\mathrm{~}\mathfrak{p}$.

Since real-world data is noisy, one certainly expects super-exponentially
growing numbers of allowed patterns in empirical observations, although
sub-factorial growths such as in Equation (\ref{Example}) seem elusive.

\section{Permutation complexity functions and classes}
\label{sec3}

Next we wish to associate the notion of permutation complexity to processes
ranging from deterministically generated signals to white noise.
Unfortunately, the metric and topological permutation entropies are not up
to the job. For instance, $h_{0}^{\ast }(\mathbf{X})$ converges for
deterministic processes (Theorem \ref{Thm1}) but diverges for
forbidden-pattern-free (FPF) signals: 
\begin{equation}
h_{0}^{\ast }(\mathbf{X})=\lim_{L\rightarrow \infty }\frac{1}{L}\ln \mathcal{%
A}_{L}(\mathbf{X})=\lim_{L\rightarrow \infty }\ln L=\infty  \label{h*_0(WN)}
\end{equation}%
by (\ref{allowed pat X}).

This being the case, we shall rather focus on the \textit{permutation
complexity} (PC) \textit{class} of the process $\mathbf{X}$, which we define
by the asymptotic growth of $\ln \mathcal{A}_{L}(\mathbf{X})$ with respect
to $L$. In view of Equation (\ref{allowed pat f}) for the deterministic
processes and Equation (\ref{allowed pat X}) for the FPF processes, we
propose the following definition.

\begin{definition}
\label{DefClass}Let $g(t)$ be a positive, invertible and sufficiently
regular function of the real variable $t\geq 0$. A process $\mathbf{X}$ is
said to belong to the PC class $g$ if%
\begin{equation}
\ln \mathcal{A}_{L}(\mathbf{X})\sim g(L)  \label{pcf}
\end{equation}%
as $L\rightarrow \infty $.
\end{definition}

The function $g(t)$ will be called the \textit{permutation complexity} (PC)%
\textit{\ function} of the process $\mathbf{X}$. The name of $g(t)$ is
suggested by Equation (\ref{allowed pat f}) with $L=\left\lfloor
t\right\rfloor $, since the topological entropy $h_{0}(f)$ measures the
dynamical complexity of the deterministic dynamic generated by $f$. In some
cases, for convenience or economy, we will group a family of classes under a
single \textquotedblleft super-class\textquotedblright , although we will
also call them classes.

\begin{remark}
\label{Remark PCF}Two important observations on the PC function of a process:

\begin{enumerate}
\item Regarding regularity, we will assume henceforth that $g(t)$ is
bicontinuous, i.e., both $g(t)$ and its inverse $g^{-1}(s)$ are continuous.
The bicontinuity and invertibility of $g(t)$ imply that $g(t)$ and, hence, $%
g^{-1}(t)$ are strictly monotonic \cite{Apostol1974}, in fact, strictly
increasing in our case.

\item Regarding uniqueness, the complexity class $g$ depends only on the
asymptotic behavior of $g(t)$; any other function $\tilde{g}(t)\sim g(t)$
(i.e., $\tilde{g}(t)=g(t)+o(g(t))$) will work out as well. Put in other
terms, PC classes are defined up to asymptotic equivalence.
\end{enumerate}
\end{remark}

Considering the growth of $\mathcal{A}_{L}(\mathbf{X})$, there is a first
clear-cut division of processes: deterministic processes, for which $%
\mathcal{A}_{L}(\mathbf{X})$ grows exponentially, and FPF processes, for
which $\mathcal{A}_{L}(\mathbf{X})$ grows factorially. Data analysis and
numerical simulations show that the latter are ubiquitous in practice.
Processes with super-exponential but sub-factorial growths will be grouped
in a third class. Specifically, we are going to turn our attention to the
following three PC classes.

\begin{description}
\item[(C1)] \textit{Exponential class}: $\ln \mathcal{A}_{L}(\mathbf{X})\sim
cL$ ($c>0$), i.e.,%
\begin{equation}
g(t)=ct=:g_{\text{exp}}(t).  \label{g_exp}
\end{equation}
\end{description}

Thus, the exponential class is actually a class of classes, one for each $c$%
. Each class with a given constant $c$ includes all deterministic processes $%
\mathbf{X}=f$ with topological entropy $h_{0}(f)=c$; maps with the same $%
h_{0}(f)$ are said to be topologically conjugate. Therefore, deterministic
processes with different topological entropies have different permutation
complexities, in line with the concept of dynamical complexity.

Moreover, for each $c>0$ the corresponding class is non-empty. Indeed, for
every $\sigma>1$ there exists a piecewise monotone map $f$ with $%
h_{0}(f)=\ln \sigma>0$, namely, the piecewise linear selfmap of the interval 
$[0,1]$ with constant slopes $\pm \sigma$. Therefore, any function of the
form $g(t)=ct$ is the PC function of a deterministic processes generated by
a piecewise linear map with $\sigma=e^{c}$.

\begin{description}
\item[(C2)] \textit{Factorial class}: $\ln \mathcal{A}_{L}(\mathbf{X})\sim
L\ln L$, i.e., 
\begin{equation}
g(t)=t\ln t=:g_{\text{fac}}(t).  \label{g_fact}
\end{equation}
\end{description}

Regarding the applications, the factorial class is the most interesting
since virtually all random processes in practice are FPF.

\begin{description}
\item[(C3)] \textit{Sub-factorial class}: $\ln \mathcal{A}_{L}(\mathbf{X}%
)\sim g(t)$, where (i) $g_{\text{exp}}(t)=o(g(t))$ and $g(t)=o(g_{\text{fac}%
}(t))$ or, else, (ii)%
\begin{equation}
g(t):=ct\ln t\;\text{\ with\ \ }0<c<1.  \label{gamma_sub}
\end{equation}
\end{description}

Unlike the exponential and factorial classes, whose PC functions are defined
explicitly, the PC functions of the sub-factorial class are defined both
implicitly (condition C3(i)) and explicitly (Equation (\ref{gamma_sub})).

The sub-factorial class is also a class of classes. This class is
potentially the largest since it fills the gap between the exponential and
the factorial class, although practical examples are hard to find. Examples
of functions $g(t)$ such that $ct=o(g(t))$ and $g(t)=o(t\ln t)$ (condition
C3(i)) are 
\begin{equation}
g(t)=t\ln ^{(n)}t\;\;(n\geq 2),  \label{gamma_sub2}
\end{equation}%
where $\ln ^{(n)}t$ denotes the composition of the logarithmic function $n$
times. Toy models with PC functions of the form (\ref{gamma_sub}) were
presented in Example \ref{ExampleSubfact}. Prompted by this example, in the
forthcoming theorems we will use $g_{\text{sub}}(t):=ct\ln t$, $0<c<1$, as a
prototypical PC function of the sub-factorial class, although the other
representatives in Equation (\ref{gamma_sub2}) will also be considered
alongside.

Let us mention in passing that $g(t)=ct\ln t$ with $c>1$ is not the PC
function of any random process in the ordinal representation. However,
statistical complex systems may have such super-factorial growth rates of
the state space as the number of constituents increases \cite{JPPT2018JPA}.

Of course, the exponential and sub-factorial classes can be thought of as
refined in smaller classes whenever convenient.

To conclude this section, let us return to the asymmetry between
deterministic and FPF processes regarding their PC functions. As already
mentioned, $g_{\text{exp}}(t)=ct$ distinguishes deterministic processes from
each other up to topological conjugacy, since $c=h_{0}(f)$ in this case. On
the contrary, all FPF processes have the same PC function, namely, $g_{\text{%
fac}}(t)=t\ln t$, the reason being that $\mathcal{A}_{L}(\mathbf{X})$ counts
the number of allowed $L$-patterns for $L\gg 1$, and this number is $L!$ for
all FPF processes. The result is that $g_{\text{fac}}(t)$ is useless in
distinguishing FPF processes from each other. A possible way out of this
shortcoming is to take into account the probability distribution of the
allowed $L$-patterns, e.g., through permutation entropies tailored to each
PC class, as we do in the next section. A different approach, based on the
convergence rate of $\ln \mathcal{A}_{L}(\mathbf{X})$ to $g_{\text{fac}%
}(t)=t\ln t$\ with the length of the time series, will be presented in
Section \ref{sec52}, when discussing numerical simulations.

\section{Generalized permutation entropy}

\label{sec4}

Let $p=(p_{1},p_{2},\ldots ,p_{W})$ be a discrete probability distribution;
we denote by $\mathcal{P}_{W}$ the set of all discrete probability
distributions with $W$ entries. From the point of view of information
theory, an entropy is a positive functional $S(p)$ defined on $\cup _{W\geq
2}\mathcal{P}_{W}$ that satisfies certain properties required by Shannon 
\cite{Shannon,Shannon2} and Khinchin \cite{Khinchin} in their uniqueness
theorem for $S(p)$, and nowadays known as the \textit{Shannon-Khinchin} (SK) 
\textit{axioms}. The first three (SK) axioms are:

\begin{description}
\item[(SK1)] \textit{Continuity}: $S$ is continuous on $\mathcal{P}_{W}$ for
each $W$.

\item[(SK2)] \textit{Maximality}: For each $(p_{1},p_{2},\ldots ,p_{W})\in 
\mathcal{P}_{W}$, 
\begin{equation*}
S(p_{1},p_{2},\ldots ,p_{W})\leq S\left( \tfrac{1}{W},\tfrac{1}{W},\ldots ,%
\tfrac{1}{W}\right) .
\end{equation*}

\item[(SK3)] \textit{Expansibility}: For each $(p_{1},p_{2},\ldots
,p_{W})\in \mathcal{P}_{W}$ and $i\in \{0,1,\ldots ,n-1\}$, 
\begin{equation*}
S(p_{1},\ldots ,p_{i},0,p_{i+1},\ldots ,p_{n})=S(p_{1},p_{2},\ldots ,p_{n}).
\end{equation*}
\end{description}

If $S(p)$ satisfies (SK1)-(SK3) and a fourth axiom called \textit{%
separability} or \textit{strong additivity} (SK4), then $S(p)$ must be the 
\textit{Boltzmann-Gibbs-Shannon entropy }(usually called Shannon entropy in
information theory): 
\begin{equation}
S(p)=-k\sum_{i=1}^{W}p_{i}\ln p_{i}=:S_{BGS}(p),  \label{S_BGS}
\end{equation}%
where $k$ is an arbitrary positive constant that can be interpreted as the
freedom in the choice of the logarithm base. If, otherwise, $S(p)$ only
satisfies the first three SK axioms, then $S(p)$ is called a \textit{%
generalized entropy} and its form is only known under additional assumptions 
\cite{Amigo2018,Ilic2021}.

\begin{remark}
In the case of \textit{group entropies}, of interest in this work, the
strong additivity axiom (SK4) is replaced by the \textit{composability axiom}%
, namely, the requirement that there exists a suitable function of the form $%
\Phi (x,y)=x+y+$ higher order terms, which takes care of the composition
process of two independent systems that are described by probability
distributions. Specifically, 
\begin{equation}
S(p\times q)=\Phi (S(p),S(q)),  \label{Phi}
\end{equation}%
where $p,q$ are any two probability distributions and $p\times q$ is their
product distribution. Here $\Phi $ is supposed to satisfy three properties: 
\textit{(i)} $\Phi (x,y)=\Phi (y,x)$ (symmetry), (ii) $\Phi (x,\Phi
(y,z))=\Phi (\Phi (x,y),z)$ (associativity), and \textit{(iii)} $\Phi
(x,0)=x $ (\textit{null-composability}), which coincide with those of a
formal group law \cite{TJ2020}. Thus, a group entropy is a functional
satisfying the first three SK axioms and the composability axiom. Property (%
\ref{Phi}) is actually crucial to generalize the standard notion of entropy.
The entropies of Shannon (\ref{S_BGS}), R\'{e}nyi (\ref{Renyi}), and Tsallis 
\cite{Tsa2009} belong to this class. A multivariate extension of the notion
of group entropy has been proposed in \cite{PT2020}. An independent
axiomatic approach to composable entropies, the pseudoadditive entropies,
has been discussed in \cite{Ilic2021} (see also the references therein).
\end{remark}

As it turns out, $S_{BGS}(p)$ is not well suited to deal with the diversity
of complex systems, including the thermodynamical ones. In complexity
theory, systems are usually classified in sub-exponential, exponential and
super-exponential \textquotedblleft complexity classes\textquotedblright ,
according to the state growth rates of the states with the number of
constituents $N$. For each of such classes there is a specific group
entropy, called $Z$-entropy, that is \textit{extensive} for the systems in
the class, meaning that it is finite over uniform probability distributions
in the limit $N\rightarrow \infty $ \cite{TJ2020,PT2020}.

In this section we capitalize on the similarities between this approach and
ours to extend the concept of permutation entropy from deterministic
processes to random processes via the $Z$-entropies for the exponential,
sub-factorial and factorial complexity classes.

\subsection{Permutation entropy of finite order}

\label{sec41}

Given a probability distribution $p=(p_{1},...,p_{W})$ and $\alpha \in 
\mathbb{R}$, $\alpha >0$, the R\'{e}nyi entropy $R_{\alpha }(p)$ is defined
as \cite{Renyi1960}%
\begin{equation}
R_{\alpha }(p)=\frac{k}{1-\alpha }\ln \left(
\sum\limits_{i=1}^{W}p_{i}^{\alpha }\right)  \label{Renyi}
\end{equation}%
($k>0$) for $\alpha \neq 1$, and 
\begin{equation}
R_{1}(p):=\lim_{\alpha \rightarrow 1}R_{\alpha }(p)=-k\sum_{i=1}^{W}p_{i}\ln
p_{i}=S_{BGS}(p),  \label{R_1}
\end{equation}%
see Equation (\ref{S_BGS}). In statistical mechanics, $k=1.380649\times
10^{-23}$ JK$^{\text{-1}}$ is the Boltzmann constant; in information theory, 
$k$ is usually set equal to 1, as we do from now on.

The following definition is an adaptation to our context of the concept of $%
Z $-entropy \cite{TJ2020,PT2020}. Remember that, according to Remark \ref%
{Remark PCF} on the PC function $g(t)$ of a process, its inverse $g^{-1}(s)$
is continuous and strictly increasing.

\begin{definition}
Let $g(t)$ be the PC function of a process $\mathbf{X}$. The \emph{(metric)
permutation entropy of order} $L$ of $\mathbf{X}$ is defined as%
\begin{equation}
Z_{g,\alpha }^{\ast }(X_{0}^{L})\equiv Z_{g,\alpha }^{\ast
}(p)=g^{-1}(R_{\alpha }(p))-g^{-1}(0),  \label{Zg_entropy}
\end{equation}%
where $\alpha >0$, $p$ is the probability distribution of the ordinal $L$%
-patterns of $X_{0}^{L}=X_{0},X_{1},...,X_{L-1}$, and $R_{\alpha }(p)$ is R%
\'{e}nyi's entropy.
\end{definition}

The term $-g^{-1}(0)$ in (\ref{Zg_entropy}) ensures that $Z_{g,\alpha
}^{\ast }(X_{0}^{L})=0$ for \textit{singular} probability distributions,
i.e., when $p_{i_{0}}=1$ and $p_{i}=0$ for $i\neq i_{0}$. By the continuity
and strictly increasing monotonicity of $g^{-1}(s)$, $Z_{g,\alpha }^{\ast
}(X_{0}^{L})$ fulfills the axioms (SK1)-(SK3), i.e., $Z_{g,\alpha }^{\ast
}(X_{0}^{L})$ is a generalized entropy. In addition, $Z_{g,\alpha }^{\ast
}(X_{0}^{L})$ satisfies the composability axiom (\ref{Phi}) with $\Phi
(x,y)=\chi ^{-1}(\chi (x)+\chi (y))$, where $\chi (t)=g(t+g^{-1}(0))$ and,
hence, $\chi ^{-1}(s)=g^{-1}(s)-g^{-1}(0)$.

By its definition (and the increasing monotonicity of $g^{-1}(s)$), $%
Z_{g,\alpha }^{\ast }(X_{0}^{L})$ inherits some of the properties of $%
R_{\alpha }(p)$. For instance, $Z_{g,\alpha }^{\ast }(X_{0}^{L})$ is
monotone decreasing with respect to the parameter $\alpha $ \cite{Amigo2018},%
\begin{equation}
Z_{g,\alpha }^{\ast }(X_{0}^{L})\geq Z_{g,\beta }^{\ast }(X_{0}^{L})\;\;%
\text{for\ \ }\alpha <\beta  \label{hierarchy}
\end{equation}%
and each $L\geq 2$.

To formulate the next theorem, we need to introduce the special function $%
\mathcal{L}(x)$, by which we denote the \textit{principal} branch of the
real $W$-\textit{Lambert function}. This is a smooth function, defined as
the solution of $ye^{y}=x$, i.e., $W(x)e^{W(x)}=x$, for $x\geq -e^{-1}$. $%
\mathcal{L}(x)$ is the unique solution for $x\geq 0$, while for $-e^{-1}\leq
x<0$ there is another solution belonging to a second branch. Some basic
properties of $\mathcal{L}(x)$ are the following \cite{Olver2010}: (i) $%
\mathcal{L}(x)$ is strictly increasing and $\cap $-convex; (ii) $\mathcal{L}%
(-e^{-1})=-1$ and $\mathcal{L}(0)=0$; (iii) $\mathcal{L}(x)>0$ for $x>0$;
(iv) $\mathcal{L}(x)>1$ for $x>e$; and (v) $\mathcal{L}(x)\rightarrow \infty 
$ as $x\rightarrow \infty $. Moreover, $\mathcal{L}(x)$ satisfies the
identity 
\begin{equation}
\mathcal{L}(x\ln x)=\ln x  \label{identity}
\end{equation}%
for $x\geq e^{-1}$.

\begin{theorem}
\label{ThmZ_class}Given a process $\mathbf{X}$, let $p$ be the probability
distribution of the ordinal $L$-patterns of $X_{0}^{L}$. For the PC classes
(C1)-(C3) of Section \ref{sec3}, the following holds.

\begin{description}
\item[(a)] For $g_{\text{exp}}(t)=ct$:%
\begin{equation}
Z_{g_{\text{exp}},\alpha }^{\ast }(X_{0}^{L})=\frac{1}{c}R_{\alpha }(p)=:Z_{%
\text{exp},\alpha }^{\ast }(X_{0}^{L}).  \label{Z_exp}
\end{equation}

\item[(b)] For $g_{\text{fac}}(t)=t\ln t$:%
\begin{equation}
Z_{g_{\text{fac}},\alpha }^{\ast }(X_{0}^{L})=e^{\mathcal{L}[R_{\alpha
}(p)]}-1=:Z_{\text{fac},\alpha }^{\ast }(X_{0}^{L}).  \label{Z_fac}
\end{equation}

\item[(c)] For $g_{\text{sub}}(t)=ct\ln t$ ($0<c<1$):%
\begin{equation}
Z_{g_{\text{sub}},\alpha }^{\ast }(X_{0}^{L})=e^{\mathcal{L}[R_{\alpha
}(p)/c]}-1=:Z_{\text{sub},\alpha }^{\ast }(X_{0}^{L}).  \label{Z_sub}
\end{equation}
\end{description}
\end{theorem}

\begin{proof}
Equation (\ref{Z_exp}) follows readily from the definition (\ref{Zg_entropy}%
) and 
\begin{equation*}
g_{\text{exp}}^{-1}(s)=\frac{s}{c}.
\end{equation*}

If $s=ct\ln t$ ($0<c\leq 1$), then $\mathcal{L}(s/c)=\mathcal{L}(t\ln t)=\ln
t$ by the identity (\ref{identity}). Hence%
\begin{equation*}
g_{\text{fac}}^{-1}(s)=e^{\mathcal{L}(s)}\text{ }
\end{equation*}%
for $c=1$, 
\begin{equation*}
g_{\text{sub}}^{-1}(s)=e^{\mathcal{L}(s/c)}
\end{equation*}%
for $c<1$, and $g_{\text{fac}}^{-1}(0)=g_{\text{sub}}^{-1}(0)=e^{\mathcal{L}%
(0)}=1$. Equations (\ref{Z_fac}) and (\ref{Z_sub}) follow.
\end{proof}

\begin{remark}
Regarding Theorem \ref{ThmZ_class}, let us highlight the following points.

\begin{enumerate}
\item $Z_{\text{exp},\alpha }^{\ast }(p)=R_{\alpha }(p)$ for $c=1$, that is,
the sub-class that includes the maps with topological entropy $1$. Since $%
R_{\alpha }(p)$ is defined anyway up to a positive constant $k$, see
Equations (\ref{Renyi})-(\ref{R_1}), we may conclude that R\'{e}nyi's
entropy (of the probability distribution of the ordinal $L$-patterns) is the
permutation entropy when dealing with deterministic processes, regardless of
their topological entropy.

\item In particular (see Equation (\ref{R_1})),%
\begin{equation}
Z_{\text{exp},1}^{\ast }(X_{0}^{L})=S_{BGS}(p)=H^{\ast }(X_{0}^{L}),
\label{Remark82}
\end{equation}%
where $H^{\ast }(X_{0}^{L})$ is the conventional metric permutation entropy
of order $L$, Equation (\ref{h*_mu,L}). In other words, $Z_{g,\alpha }^{\ast
}(X_{0}^{L})$ reduces to the conventional permutation entropy under the
right assumptions. This justifies calling it a (generalized) permutation
entropy.

\item $Z_{\text{fac},\alpha }^{\ast }(X_{0}^{L})$ was used in \cite%
{Amigo2021} (with the notation $Z_{\alpha }^{\ast }(X_{0}^{L})$) to
generalize $H^{\ast }(X_{0}^{L})$ to FPF processes. There it is proved that%
\begin{equation}
Z_{\text{fac},\alpha }^{\ast }(p)=R_{\alpha }(p)-\frac{1}{2}R_{\alpha
}(p)^{2}+O(3)  \label{Remark83}
\end{equation}%
if $R_{\alpha }(p)<1/e$. Therefore, when $R_{\alpha }(p)$ is small, it is a
good approximation of $Z_{\text{fac},\alpha }^{\ast }(p)$.
\end{enumerate}
\end{remark}

As anticipated in Section \ref{sec3}, we have chosen $g_{\text{sub}%
}(t)=ct\ln t$ ($0<c<1$) in Theorem \ref{ThmZ_class}(c) mainly because of
Example \ref{ExampleSubfact}. For the choice $g_{\text{sub}}(t)=t\ln ^{(n)}t$
($n\geq 2$), the other examples of sub-factorial PC functions given in
Equation (\ref{gamma_sub2}), we need to generalize the Lambert function $%
\mathcal{L}(x)$. We define the \textit{generalized Lambert function} $%
\mathcal{L}^{(n)}(x)$ ($n\geq 1$, with $\mathcal{L}^{(1)}(x)=\mathcal{L}(x)$%
) by the functional equation%
\begin{equation}
\mathcal{L}^{(n)}(x)\exp ^{(n)}[\mathcal{L}^{(n)}(x)]=x  \label{Lambert(n)}
\end{equation}%
for $x\geq -\exp ^{(n)}(-1)$, where $\exp ^{(n)}(x)$ denotes the composition
of the exponential function $n$ times. Hence, $\mathcal{L}^{(n)}(x)\geq 0$
for $x\geq 0$, $\mathcal{L}^{(n)}(0)=0$, and the identity (\ref{identity})
generalizes to%
\begin{equation}
\mathcal{L}^{(n)}[x\ln ^{(n)}x]=\ln ^{(n)}x  \label{identity2}
\end{equation}%
for $x\geq \exp ^{(n)}(-1)$ (since $\mathcal{L}^{(n)}[-\exp ^{(n)}(-1)]=-1$%
). It follows that the inverse of $g(t)=t\ln ^{(n)}t$ is%
\begin{equation}
g^{-1}(s)=\exp ^{(n)}[\mathcal{L}^{(n)}(s)],  \label{inv_sub}
\end{equation}%
so that $g^{-1}(0)=1$ since $\mathcal{L}^{(n)}(0)=0$. Therefore, the
permutation entropy of order $L$ of the sub-factorial class defined by $%
g(t)=t\ln ^{(n)}t$ is 
\begin{equation}
Z_{g,\alpha }^{\ast }(X_{0}^{L})=\exp ^{(n)}[\mathcal{L}^{(n)}(R_{\alpha
}(p))]-1\;\;(n\geq 2).  \label{PE_sub}
\end{equation}%
For $n=1$ we recover $Z_{\text{fac},\alpha }^{\ast }(X_{0}^{L})$, Equation (%
\ref{Z_fac}).

\subsection{Permutation entropy rate}

\label{sec42}

According to axiom SK2, entropies reach their maxima over uniform
probability distributions. Sometimes these upper bounds are called the
topological versions of the corresponding entropies or simply topological
entropies. Thus, the topological version of $Z_{g,\alpha }^{\ast
}(X_{0}^{L}) $ is its tight upper bound, which is obtained over the uniform
distribution $p_{u}$ of the \textit{allowed} ordinal $L$-patterns for $%
\mathbf{X}$. This means that $p_{u}=(p_{1},...,p_{L!})$ with%
\begin{equation}
p_{i}=\left\{ 
\begin{array}{cl}
1/\mathcal{A}_{L}(\mathbf{X}) & \text{if the }i\text{th ordinal }L\text{%
-pattern is allowed for }\mathbf{X}\text{ } \\ 
0 & \text{if the }i\text{th ordinal }L\text{-pattern is forbidden for }%
\mathbf{X}%
\end{array}%
\right.  \label{p_u}
\end{equation}%
for $i=1,...,L!$. Note that%
\begin{equation}
R_{\alpha }(p_{u})=\ln \mathcal{A}_{L}(\mathbf{X})  \label{R(p_u)}
\end{equation}%
for $\alpha >0$. Plugging Equation (\ref{R(p_u)}) into (\ref{Zg_entropy}),
we are led to the following definition.

\begin{definition}
The \emph{topological permutation entropy of order} $L$ of a process $%
\mathbf{X}$ of class $g$ is defined as 
\begin{equation}
Z_{g,0}^{\ast }(X_{0}^{L})\equiv Z_{g,0}^{\ast }(p_{u})=g^{-1}(\ln \mathcal{A%
}_{L}(\mathbf{X}))-g^{-1}(0),  \label{Z*_g,0}
\end{equation}%
where $p_{u}$ is the uniform probability distribution of allowed $L$%
-patterns for $\mathbf{X}$ as defined in Equation (\ref{p_u}).
\end{definition}

The notation $Z_{g,0}^{\ast }$ for the topological permutation entropy is
justified because $\ln \mathcal{A}_{L}(\mathbf{X})$ is formally obtained
from Equation (\ref{Renyi}) by setting $\alpha =0$; indeed,%
\begin{equation}
R_{0}(p_{1},...,p_{L!})=\ln \left\vert \{p_{i}:p_{i}>0,\,1\leq i\leq L!\text{
}\}\right\vert =\ln \mathcal{A}_{L}(\mathbf{X})  \label{R_0}
\end{equation}%
for all $p=\{p_{1},...,p_{L!}\}$. It follows,%
\begin{equation}
R_{0}(p)\geq R_{\alpha }(p)  \label{R_0>R_alpha}
\end{equation}%
for all $\alpha >0$, so that 
\begin{equation}
Z_{g,0}^{\ast }(X_{0}^{L})\geq Z_{g,\alpha }^{\ast }(X_{0}^{L})
\label{hierarch}
\end{equation}%
for all $\alpha >0$ since $g^{-1}(s)$ is a strictly increasing function (see
Remark \ref{Remark PCF}).

Uniform probability distributions are special for several reasons. From the
viewpoint of statistical mechanics, they correspond to the most disordered
state, hence to equilibrium in the microcanonical ensemble. From the point
of view of information theory, they amount to the principle of insufficient
reason or maximum entropy principle \cite{Jaynes1957} under null knowledge
(maximum ignorance). Most important for us, the concept of extensivity
(inherited from thermodynamics) also refers to such probability
distributions: we say that an entropy $S(p)$ is \textit{extensive} if it
scales linearly with the number of constituents (degrees of freedom, etc.) $%
N $ of the system over the uniform probability distribution $p=(1/N,...,1/N)$%
, i.e.,%
\begin{equation}
\lim_{N\rightarrow \infty }\frac{S(\frac{1}{N},...,\frac{1}{N})}{N}\sim
~const>0.  \label{extensiv}
\end{equation}%
Therefore, extensivity depends on how the number of states grows with $N$,
the sub-exponential, exponential and super-exponential regimes (or classes)
being the most important ones.

\begin{theorem}
\label{ThmExtensivity}The permutation entropy $Z_{g,\alpha }^{\ast
}(X_{0}^{L})$ is extensive with respect to the parameter $L$. In fact, for
all $\alpha >0$,%
\begin{equation}
\frac{Z_{g,\alpha }^{\ast }(p_{u})}{L}=\frac{Z_{g,0}^{\ast }(X_{0}^{L})}{L}%
\sim 1.  \label{extensivity Z}
\end{equation}
\end{theorem}

\begin{proof}
From definition (\ref{Z*_g,0}) and $\ln \mathcal{A}_{L}(\mathbf{X})\sim g(L)$%
, Equation (\ref{pcf}), we obtain%
\begin{equation*}
Z_{g,0}^{\ast }(X_{0}^{L})=g^{-1}(\ln \mathcal{A}_{L}(\mathbf{X}%
))-g^{-1}(0)\sim g^{-1}(g(L))-g^{-1}(0)=L-g^{-1}(0)\sim L.
\end{equation*}%
This proves Equation (\ref{extensivity Z}).
\end{proof}

To get rid of the dependence of $Z_{g,\alpha }^{\ast }(X_{0}^{L})$ on $L$,
we turn to the entropy rates per variable, $Z_{g,\alpha }^{\ast
}(X_{0}^{L})/L$, and take the limit when $L\rightarrow \infty $.

\begin{definition}
The \emph{permutation entropy rate}\textit{\ }(or just \emph{permutation
entropy})\textit{\ of a process} $\mathbf{X}$ of class $g$ is defined as 
\begin{equation}
z_{g,\alpha }^{\ast }(\mathbf{X})=\lim_{L\rightarrow \infty }\frac{1}{L}%
Z_{g,\alpha }^{\ast }(X_{0}^{L}),  \label{Z*}
\end{equation}%
where $\alpha \geq 0$: $z_{0}^{\ast }(\mathbf{X})$ is the \emph{topological}
permutation entropy of $\boldsymbol{X}$, and $z_{\alpha }^{\ast }(\mathbf{X}%
) $ with $\alpha >0$ is the $\emph{metric}$ permutation entropy of $\mathbf{X%
}$.
\end{definition}

The permutation entropy rate $z_{g,\alpha }^{\ast }(\mathbf{X})$ quantifies
an intrinsic property of the process $\mathbf{X}$. The existence of the
limit (\ref{Z*}) follows from Theorem \ref{ThmExtensivity}. As a matter of
fact, the existence of $z_{g,\alpha }^{\ast }(\mathbf{X})$ amounts to the
extensivity of $Z_{g,\alpha }^{\ast }(X_{0}^{L})$.

\begin{theorem}
For each complexity class $g$ and $\alpha >0$,%
\begin{equation}
z_{g,\alpha }^{\ast }(\mathbf{X})\leq z_{g,0}^{\ast }(\mathbf{X})=1
\label{bounds2}
\end{equation}
\end{theorem}

\begin{proof}
The inequality $z_{g,\alpha }^{\ast }(\mathbf{X})\leq z_{g,0}^{\ast }(%
\mathbf{X})$ follows from Equation (\ref{hierarch}). The equality%
\begin{equation}
z_{g,0}^{\ast }(\mathbf{X})=\lim_{L\rightarrow \infty }\frac{1}{L}%
Z_{g,0}^{\ast }(X_{0}^{L})=1  \label{z*_0}
\end{equation}%
follows from Equation (\ref{extensivity Z}).
\end{proof}

Therefore, the values of the entropy rate $z_{g,\alpha }^{\ast }(\mathbf{X})$%
, $\alpha \geq 0$, are restricted to the unit interval $[0,1]$. The next
theorem gives $z_{g,\alpha }^{\ast }(\mathbf{X})$ for the exponential and
factorial classes, along with the sub-factorial class for $g(t)=ct\ln t$ ($%
0<c<1$).

\begin{theorem}
\label{Thm_h*}Given a process $\mathbf{X}$, let $p$ be the probability
distribution of the ordinal $L$-patterns of $X_{0}^{L}$. The permutation
entropy rate of $\mathbf{X}$ is given as follows.

\begin{description}
\item[(a)] For the exponential class:%
\begin{equation}
z_{\text{exp},\alpha }^{\ast }(\mathbf{X}):=z_{g_{\text{exp}},\alpha }^{\ast
}(\mathbf{X})=\lim_{L\rightarrow \infty }\frac{1}{cL}R_{\alpha }(p).
\label{z*_exp}
\end{equation}%
In particular, if $\mathbf{X}=f$ and $\alpha =1$, then%
\begin{equation}
z_{\text{exp},1}^{\ast }(\mathbf{X})=\frac{h(f)}{h_{0}(f)},  \label{z*_expB}
\end{equation}%
where $h(f)$ is the Kolmogorov-Sinai entropy of $f$ and $h_{0}(f)$ is its
topological entropy.

\item[(b)] For the factorial class:%
\begin{equation}
z_{\text{fac},\alpha }^{\ast }(\mathbf{X}):=z_{g_{\text{fac}},\alpha }^{\ast
}(\mathbf{X})=\lim_{L\rightarrow \infty }\frac{1}{L}e^{\mathcal{L}[R_{\alpha
}(p)]}.  \label{z*_fac}
\end{equation}

\item[(c)] For the sub-factorial class (defined by $g_{\text{sub}}(t)=ct\ln
t $, $0<c<1$):%
\begin{equation}
z_{\text{sub},\alpha }^{\ast }(\mathbf{X}):=z_{g_{\text{sub}},\alpha }^{\ast
}(\mathbf{X})=\lim_{L\rightarrow \infty }\frac{1}{L}e^{\mathcal{L}[R_{\alpha
}(p)/c]}.  \label{z*_sub}
\end{equation}
\end{description}
\end{theorem}

\begin{proof}
Use the definition (\ref{Z*}) and 
\begin{equation*}
Z_{g,\alpha }^{\ast }(X_{0}^{L})\sim Z_{g_{\text{exp}},\alpha }^{\ast
}(X_{0}^{L})=\frac{1}{c}R_{\alpha }(p)
\end{equation*}%
for $\mathbf{X}$ in the exponential class (Theorem \ref{ThmZ_class}(a)),%
\begin{equation*}
Z_{g,\alpha }^{\ast }(X_{0}^{L})\sim Z_{g_{\text{fac}},\alpha }^{\ast
}(X_{0}^{L})=e^{\mathcal{L}[R_{\alpha }(p)]}-1
\end{equation*}%
for $\mathbf{X}$ in the factorial class (Theorem \ref{ThmZ_class}(b)), and%
\begin{equation*}
Z_{g,\alpha }^{\ast }(X_{0}^{L})\sim Z_{g_{\text{sub}},\alpha }^{\ast
}(X_{0}^{L})=e^{\mathcal{L}[R_{\alpha }(p)/c]}-1
\end{equation*}%
for $\mathbf{X}$ in the sub-factorial class defined by $g_{\text{sub}%
}(t)=ct\ln t$, $0<c<1$ (Theorem \ref{ThmZ_class}(c)), to derive Equations (%
\ref{z*_exp}), (\ref{z*_fac}) and (\ref{z*_sub}), respectively.

As for Equation (\ref{z*_expB}), use (i) $R_{1}(p)=S_{BGS}(p)=H^{\ast
}(X_{0}^{L})$, (ii)%
\begin{equation*}
\lim_{L\rightarrow \infty }\frac{1}{L}R_{1}(p)=\lim_{L\rightarrow \infty }%
\frac{1}{L}H^{\ast }(X_{0}^{L})=h(f)
\end{equation*}%
(Theorem \ref{Thm1}(a)), and (iii) $g_{\text{exp}}(t)=ct=h_{0}(f)t$ in
Equation (\ref{z*_exp}).
\end{proof}

Note that Equation (\ref{z*_expB}) complies with (\ref{bounds2}) for $g=g_{%
\text{exp}}$ and $\alpha =1$ because $h_{0}(f)\geq h(f).$

\begin{remark}
\label{RemarkLast}A few closing observations on permutation entropy rates.

\begin{enumerate}
\item The hierarchical order (\ref{hierarchy}) and (\ref{hierarch}) carries
on trivially to permutation entropy rates:%
\begin{equation}
z_{g,\alpha }^{\ast }(\mathbf{X})\geq z_{g,\beta }^{\ast }(\mathbf{X})
\label{Remark15a}
\end{equation}%
for $0\leq \alpha <\beta $.

\item For the sub-factorial subclasses defined by $g(t)=t\ln ^{(n)}t$, $%
n\geq 2$, use Equation (\ref{PE_sub}) to obtain%
\begin{equation}
z_{g,\alpha }^{\ast }(\mathbf{X})=\lim_{L\rightarrow \infty }\frac{1}{L}\exp
^{(n)}[\mathcal{L}^{(n)}[R_{\alpha }(p)]].  \label{Remark15b}
\end{equation}

\item For white noise (WN), the probability distribution of the $L$-patterns
is uniform for every $L$. Therefore, 
\begin{equation}
z_{\text{fac},\alpha }^{\ast }(\text{WN})=z_{\text{fac},0}^{\ast }(\text{WN}%
)=1  \label{Remark15c}
\end{equation}%
for every $\alpha >0$ by Equation (\ref{z*_0}).
\end{enumerate}
\end{remark}

\section{Numerical simulations}

\label{sec5}

One of the most important applications of permutation complexity to data
analysis is the characterization and, hence, discrimination of time series.
In this section we illustrate the discrimination power of permutation
complexity in time series analysis. For this purpose, we are going to use
permutation entropy in Section \ref{sec51} and (perhaps surprisingly) the PC
function in Section \ref{sec52}. Both tools are used in Section \ref{sec53}
to further dissect the sub-factorial \textquotedblleft not-so-noisy
measurement of a periodic signal\textquotedblright\ introduced in Example %
\ref{ExampleSubfact}. Numerical receipts for computing ordinal patterns and
permutation entropies can be found, e.g., in References \cite%
{Pessa2021,Unakafova2013}.

Since real time series analysis have finite length, some allowed ordinal
patterns can be missing in random time series simply for statistical
reasons. Therefore, practitioners prefer to speak of \textit{visible patterns%
} and \textit{missing patterns} rather than allowed patterns and forbidden
patterns, respectively, as we will sometimes do as well.

\subsection{Time series discrimination using permutation entropies}

\label{sec51}

In Section \ref{sec4} we have explicitly constructed a permutation entropy $%                                                                        
Z_{g,\alpha }^{\ast }(X_{0}^{L})$ for each PC class $g$ such that the                                                                                
corresponding entropy rate $z_{g,\alpha }(\mathbf{X})$ is finite. However,                                                                           
real-world data is noisy, which seems to exclude the exponential class                                                                               
---but not quite.                                                                                                                                    
                                                                                                                                                     
In nonlinear time series analysis, it is good practice to test for                                                                                   
determinism first. Underlying determinism in noisy time series can be                                                                                
unveiled by several techniques \cite{Kantz2004}, including forbidden ordinal                                                                         
patterns \cite{Amigo2010B}. If the noise to signal ratio is sufficiently                                                                             
small, then the data can be denoised, which allows the analyst to work with                                                                          
time series as good as noiseless deterministic. This is the exponential                                                                              
class, and the realm of the conventional permutation entropy $Z_{\text{exp}%                                                                         
,\alpha }^{\ast }(X_{0}^{L})=R_{\alpha }(p)$ (or $H^{\ast }(X_{0}^{L})$ for $%                                                                       
\alpha =1$) and its rate $z_{\text{exp},\alpha }^{\ast }(\mathbf{X})$. Since                                                                         
real-world time series are finite, the entropy rate $z_{\text{exp},\alpha                                                                            
}^{\ast }(\mathbf{X})$ can only be estimated if the convergence of $Z_{\text{%                                                                       
exp},\alpha }^{\ast }(X_{0}^{L})/L$ is sufficiently fast. This can be                                                                                
checked, e.g., by plotting $Z_{\text{exp},\alpha }^{\ast                                                                                             
}(X_{0}^{L})/L=R_{\alpha }(p)/L$ vs $1/L$; if there is an interval where the                                                                         
curve is linear (before undersampling sets in), then fit a straight line to                                                                          
the linear segment of the curve and the sought limit $z_{\text{exp},\alpha                                                                           
}^{\ast }(\mathbf{X})=\lim_{L\rightarrow \infty }Z_{\text{exp},\alpha                                                                                
}^{\ast }(X_{0}^{L})/L$ is the intercept of the straight line with the                                                                               
vertical axis \cite[Sect. 2.1]{Amigo2010}. If desired, the parameter $c$                                                                             
that appears in Equation (\ref{Z_exp}) can be estimated by $H_{0}^{\ast                                                                              
}(X_{0}^{L})/L=\ln \mathcal{A}_{L}(\mathbf{X})/L$ (see Equations (\ref{h*_0}%                                                                        
) and (\ref{h*_0,L})), because $H_{0}^{\ast }(X_{0}^{L})/L$ is a proxy of $%                                                                         
c=h_{0}(f)$ for $L$ large enough by Theorem \ref{Thm1}(b) with $f=\mathbf{X}$%                                                                       
. For the purpose of time series discrimination, however, $c$ can be                                                                                 
dispensed with, which amounts to setting $c=1$. The estimation of $\mathcal{A%                                                                       
}_{L}(\mathbf{X})$ is usually done by just counting visible patterns in a                                                                            
sample of time series or even in a single, sufficiently long time series.                                                                            
This procedure can be justified if the orbits densely visit\ the state                                                                               
space, a property that goes by the name of transitivity. By the way, this is                                                                         
the first property in Devaney's definition of chaos and, in fact, it implies                                                                         
the other two properties (density of periodic points and sensitivity to                                                                              
initial conditions) for interval maps \cite{Ruette2017}.                                                                                             
                                                                                                                                                     
Furthermore, in nonlinear time series analysis, the dynamics of the (often                                                                           
unknown) system under observation is assumed to settle down on a low                                                                                 
dimensional attractor, where it is transitive. However, for the asymptotic                                                                           
dynamics to be accessible to finite precision observations and numerical                                                                             
simulations, it is necessary that the physical measure (see Equation (6)) is                                                                         
smooth\ (or absolutely continuous in technical terminology \cite{Walter2000}%                                                                        
). Typically, the physical measure of chaotic attractors has a smooth                                                                                
density in the stretching, or unstable, directions of the dynamics, while it                                                                         
has a discontinuous (e.g., Cantor set-like) structure transversally to those                                                                         
directions \cite{Eckmann1985}; think of the H\'{e}non attractor. Finite                                                                              
precision smooths out the physical measure when the attractor is viewed                                                                              
transversally to the stretching directions.                                                                                                          
                                                                                                                                                     
Otherwise, if there is no good reason to assume determinism, the data are                                                                            
handled as random. Moreover, numerical simulations and empirical                                                                                     
observations show that virtually all random time series encountered in                                                                               
practice are FPF. This entails that $Z_{\text{fac},\alpha }^{\ast                                                                                    
}(X_{0}^{L})$ and $z_{\text{fac},\alpha }^{\ast }(\mathbf{X})$ are the                                                                               
appropriate tools to characterize random time series in the absence of more                                                                          
information. However, the factorial growth of the $L$-patterns and the                                                                               
associated computational cost restrict $L$ to moderate values in practice ($%                                                                        
L\lesssim 7$) , which makes the numerical estimation of $z_{\text{fac}%                                                                              
,\alpha }^{\ast }(\mathbf{X})$ an open question in general.                                                                                          
                                                                                                                                                     
For the above reasons, we have selected seven random processes from the                                                                              
factorial class to illustrate the discrimination power of permutation                                                                                
complexity with numerical simulations. This is also a particularly difficult                                                                         
case because all processes belong to the same complexity class in strict                                                                             
sense (there are no parameters in $g_{\text{fac}}(t)=t\ln t$). Those seven                                                                           
processes are the following.                                                                                                                         
                                                                                                                                                     
\begin{description}                                                                                                                                  
\item[(Fac1)] \textit{White noise} (WN) in the form of an independent and                                                                            
uniformly distributed process on $[0,1]$;                                                                                                            
                                                                                                                                                     
\item[(Fac2)] \textit{Fractional Gaussian noise} (fGn) with Hurst exponent $%                                                                        
H=0.2$ \cite{Mandelbrot1968};                                                                                                                        
                                                                                                                                                     
\item[(Fac3-5)] \textit{Fractional Brownian motion} (fBm) with $H=0.2$, $0.4$                                                                        
(anti-persistent processes), and $H=0.6$ (persistent process) \cite%                                                                                 
{Mandelbrot1968};                                                                                                                                    
                                                                                                                                                     
\item[(Fac6)] \textit{Logistic map} \textit{with additive white noise} of                                                                            
amplitude $0.30$ (noisy LM), i.e., $x_{t}=y_{t}+z_{t}$, where $%                                                                                     
y_{t}=4y_{t-1}(1-y_{t-1})$ with $y_{0}=0.2002$, and $(z_{t})_{t\geq 0}$ is                                                                           
WN with $-0.30\leq z_{t}\leq 0.30$;                                                                                                                  
                                                                                                                                                     
\item[(Fac7)] \textit{Schuster map }with exponent 2 \cite{Schuster1988} and%                                                                         
\textit{\ additive white noise} of amplitude $0.25$ (noisy SM), i.e., $%                                                                             
x_{t}=y_{t}+z_{t}$, where $y_{t}=y_{t-1}+y_{t-1}^{2}~\mathrm{{mod}~1}$ with $%                                                                       
y_{0}=0.2002$, and $(z_{t})_{t\geq 0}$ is WN with $-0.25\leq z_{t}\leq 0.25$.                                                                        
\end{description}                                                                                                                                    
                                                                                                                                                     
Arguments for this specific pick include that (i) the processes                                                                                      
(Fac1)-(Fac7) cover a diversity of interesting cases (white noise, random                                                                            
processes with long dependence ranges, deterministic dynamics contaminated                                                                           
with observational noise); (ii) they are relevant to time series analysis                                                                            
and familiar to the analysts; and (iii) there are well-tested numerical                                                                              
routines available for simulations.                                                                                                                  
                                                                                                                                                     
Figure~\ref{fig:Figura1} shows $\langle Z_{\text{fac},\alpha }^{\ast                                                                                 
}(X_{0}^{L})/L\rangle $, the average of $Z_{\text{fac},\alpha }^{\ast                                                                                
}(X_{0}^{L})/L$ over 35 time series for the random processes (Fac1)-(Fac7)                                                                           
and $3\leq L\leq 7$, where $\alpha =0.5$ (a), $\alpha =1$ (b) and $\alpha                                                                            
=1.5$ (c). For calculation purposes, the maximal length of the time series                                                                           
was set at $T_{\max }=50$,$000$ ($\gg 7!=5040$), but the computational loop                                                                          
is actually exited as soon as the probability distribution of the $L$%                                                                               
-patterns stabilizes, so that the numerical routine is the same for all $%                                                                           
3\leq L\leq 7$. We see in all panels of Figure~\ref{fig:Figura1} that $%                                                                             
\langle Z_{\text{fac},\alpha }^{\ast }(X_{0}^{L})/L\rangle $ follows a                                                                               
distinct and seemingly convergent trajectory for each process as $L$ grows,                                                                          
upper bounded by the white noise. In agreement with (\ref{hierarchy}), $%                                                                            
\langle Z_{\text{fac},0.5}^{\ast }(X_{0}^{L})/L\rangle \geq \langle Z_{\text{%                                                                       
fac},1}^{\ast }(X_{0}^{L})/L\rangle \geq \langle Z_{\text{fac},1.5}^{\ast                                                                            
}(X_{0}^{L})/L\rangle $\ for each process. For white noise, $\langle Z_{%                                                                            
\text{fac},\alpha }^{\ast }(X_{0}^{L})/L\rangle \rightarrow 1$ as $L$ grows                                                                          
by Equation (\ref{Remark15c}).                                                                                                                       
                                                                                                                                                     
The effect of $\alpha $ on the discriminatory power of $Z_{\text{fac},\alpha                                                                         
}(X_{0}^{L})$ clearly depends on the probability distribution $p$ of the $L$%                                                                        
-patterns through the R\'{e}nyi entropy $R_{\alpha }(p)$. In particular, for                                                                         
$\alpha <1$ the central part of the distribution is flattened, i.e.,                                                                                 
high-probability events are suppressed, and low-probability events are                                                                               
enhanced. This effect is more pronounced for smaller $\alpha $. The opposite                                                                         
happens when $\alpha >1$: low-probability events are suppressed while                                                                                
high-probability events are enhanced. Our choices $\alpha =0$.$5$, $1$, $1$.$%                                                                       
5$ are meant to include both situations $\alpha <1$ and $\alpha >1$, along                                                                           
with the Shannonian case $\alpha =1$. In Figure \ref{fig:Figura1}, this                                                                              
results in a higher discriminatory power of $Z_{\text{fac},\alpha                                                                                    
}(X_{0}^{L})$ with increasing $\alpha $, i.e., the larger $\alpha $ the                                                                              
further apart the curves $\langle Z_{\text{fac},\alpha }^{\ast                                                                                       
}(X_{0}^{L})/L\rangle $ are, which shows that the parameter $\alpha $ is an                                                                          
asset in applications. Similar choices of $\alpha $ give similar results                                                                             
(not shown).                                                                                                                                         
                                                                                                                                                     
Note also that the curves of different processes may cross. The reason is                                                                            
that $Z_{\text{fac},\alpha }^{\ast }(X_{0}^{L})$ can only capture ranges of                                                                          
interdependence up to $L$. Put another way, larger window sizes\ $L$ unveil                                                                          
dependencies between farther variables that can be measured by $Z_{\text{fac}%                                                                       
,\alpha }^{\ast }(X_{0}^{L})$. Therefore, as $L$ grows, $Z_{\text{fac}%                                                                              
,\alpha }^{\ast }(X_{0}^{L})$ can become larger for a noisy chaotic signal,                                                                          
such as the noisy logistic map, than for a process with a longer, or an                                                                              
infinite, span of interdependence between its increments, such as the                                                                                
fractional Brownian motion with $H=0.40$; this can be see more clearly in                                                                            
panel (c).                                                                                                                                           
                                                                                                                                                     
In particular, Figure~\ref{fig:Figura1} shows that, although all seven                                                                               
processes (Fac1)-(Fac7) belong to the same PC class, namely, the factorial                                                                           
class, the finite rates $Z_{\text{fac},\alpha }^{\ast }(X_{0}^{L})/L$ can                                                                            
distinguish them, evidencing as expected that $Z_{\text{fac},\alpha }^{\ast                                                                          
}(X_{0}^{L})$ is a finer measure of permutation complexity than $g(t)$. This                                                                         
does not mean that the growth rate of allowed patterns cannot be utilized                                                                            
for that objective, as we explain in Section \ref{sec52}.                                                                                            
                                                                                                                                                     
\begin{figure*}[tbp]                                                                                                                                 
\begin{center}                                                                                                                                       
\includegraphics[scale=0.45]{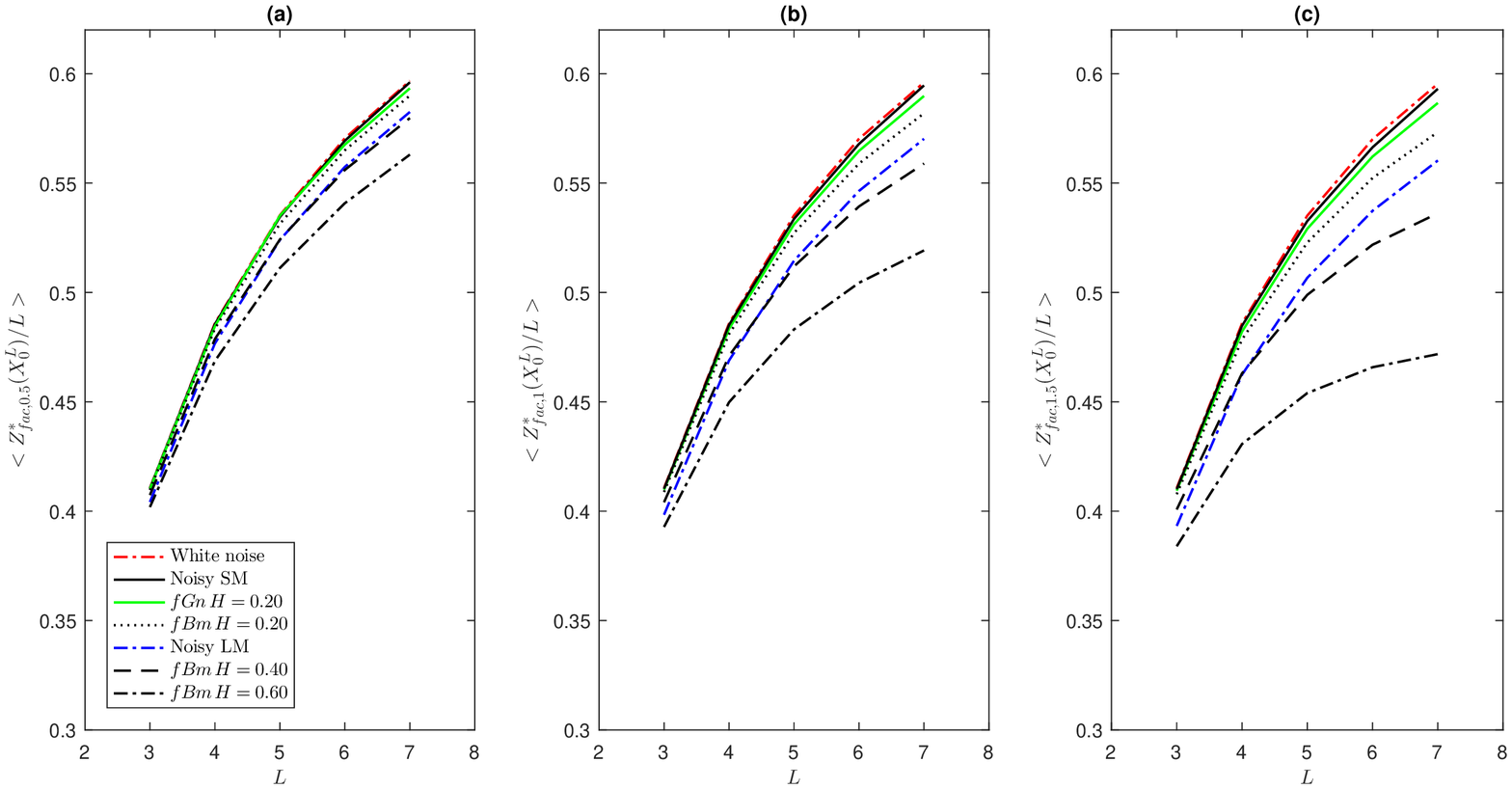}                                                                                                                
\end{center}                                                                                                                                         
\caption{ The averages $\langle Z_{\text{fac},\protect\alpha }^{\ast                                                                                 
}(X_{0}^{L})/L\rangle $ of $Z_{\text{fac},0.5}^{\ast }(X_{0}^{L})/L$ (a), $%                                                                         
Z_{\text{fac},1}^{\ast }(X_{0}^{L})/L$ (b) and $Z_{\text{fac},1.5}^{\ast                                                                             
}(X_{0}^{L})/L$ (c) over 35 realizations of the random processes listed in                                                                           
the inset are plotted vs $L$ for $3\leq L\leq 7$. }                                                                                                  
\label{fig:Figura1}                                                                                                                                  
\end{figure*}                                                                                                                                        
                                                                                                                                                     
\subsection{Time series discrimination using permutation complexity functions%                                                                       
}                                                                                                                                                    
                                                                                                                                                     
\label{sec52}                                                                                                                                        
                                                                                                                                                     
It was noticed in \cite{Amigo2007} that the number of missing $L$-patterns                                                                           
for a \textit{white noise} series of length $T\geq L$ decreases                                                                                      
exponentially with $T$. This result was generalized in \cite{Carpi2010} to                                                                           
different sorts of \textit{FPF processes} ($f^{-k}$ power spectrum (PS),                                                                             
fractional Brownian motion (fBm), fractional Gaussian noise (fGn)) by setting%                                                                       
\begin{equation}                                                                                                                                     
\mathcal{M}_{L,T}(\mathbf{X}):=\left\vert \{\text{missing }L\text{-patterns                                                                          
in }x_{0}^{T-1}\}\right\vert =Ce^{-RT},  \label{decay rate}                                                                                          
\end{equation}%                                                                                                                                      
where $x_{0}^{T-1}$ is a typical realization of $\mathbf{X}$ of length $%                                                                            
T\geq L$, and $C$ and the \textit{decay rate} $R$ are constants that depend                                                                          
on $L$ and the parameters of the random process ($k$ for $f^{-k}$ PS, or the                                                                         
Hurst exponent $H$ for BM and fGn). From Equation (\ref{decay rate}) and $%                                                                          
\mathcal{M}_{L,L}(\mathbf{X})=L!-1$, it follows $C=(L!-1)e^{RT}$, hence%                                                                             
\begin{equation}                                                                                                                                     
\mathcal{M}_{L,T}(\mathbf{X})=(L!-1)e^{-R(T-L)}.  \label{decay rate2}                                                                                
\end{equation}%                                                                                                                                      
Clearly, the decay rate of missing patterns in random time series is not                                                                             
indifferent to the dependencies between the variables of the process, either                                                                         
due to an underlying functional dependence (noisy deterministic signal) or                                                                           
to a statistical correlation.                                                                                                                        
                                                                                                                                                     
This being the case, we let the permutation complexity function $g$ depend                                                                           
on $T$ too, i.e., $g=g(L,T)$, and generalize Equation (\ref{pcf}) to%                                                                                
\begin{equation}                                                                                                                                     
\ln \mathcal{A}_{L,T}(\mathbf{X}):=\ln \left\vert \{\text{visible }L\text{%                                                                          
-patterns in }x_{0}^{T-1}\}\right\vert =g(L,T)  \label{A(L,N)}                                                                                       
\end{equation}%                                                                                                                                      
so that $g(L)\sim g(L,T)$, where now $L\gg 1$ implies $T\gg 1$ because $%                                                                            
T\geq L$. Therefore,%                                                                                                                                
\begin{equation}                                                                                                                                     
\mathcal{A}_{L,T}(\mathbf{X})+\mathcal{M}_{L,T}(\mathbf{X})=L!  \label{A+M}                                                                          
\end{equation}%                                                                                                                                      
for every $2\leq L\leq T$. Since there are two parameters $L$ and $T$, we                                                                            
can fix $L$ ($=L_{0}$) and vary $T$ so that the allowed patterns have a                                                                              
chance to become visible. Correspondingly, we can distinguish (i) a \textit{%                                                                        
transient phase} ($T$ \textquotedblleft small\textquotedblright ), where the                                                                         
allowed patterns become progressively visible, and (ii) a \textit{stationary                                                                         
phase} ($T$ \textquotedblleft large\textquotedblright ), where all the                                                                               
allowed patterns are visible. Generally speaking, the transient phase                                                                                
discriminates random signals while the stationary phase discriminates                                                                                
deterministic signals, so both phases complement each other in the analysis                                                                          
of permutation complexity.                                                                                                                           
                                                                                                                                                     
Indeed, in the deterministic case $\mathbf{X}=f$,                                                                                                    
\begin{equation}                                                                                                                                     
g(L_{0},T)=\ln \mathcal{A}_{L_{0},T}(\mathbf{X})\nearrow \ln \mathcal{A}%                                                                            
_{L_{0}}(\mathbf{X})=h_{0,L_{0}}^{\ast }(f)L_{0}  \label{gamma(L,N)_det}                                                                             
\end{equation}%                                                                                                                                      
as $T$ increases, see Equations (\ref{h*_0}) and (\ref{h*_0,L}). Here, the                                                                           
symbol $\nearrow $ means that the convergence is monotone increasing. The                                                                            
limit is achieved in finite time, namely, once all allowed $L_{0}$-patterns                                                                          
are visible. For FPF processes, though,                                                                                                              
\begin{equation}                                                                                                                                     
\mathcal{A}_{L,T}(\mathbf{X})=L!-\mathcal{M}_{L,T}(\mathbf{X})\simeq                                                                                 
L!\left( 1-e^{-R(T-L)}\right)  \label{decay0}                                                                                                        
\end{equation}%                                                                                                                                      
for every $L$ by Equations (\ref{A+M}) and (\ref{decay rate2}). Instead of                                                                           
the limit (\ref{gamma(L,N)_det}) for $\mathbf{X}=f$, for FPF processes we                                                                            
have                                                                                                                                                 
\begin{equation}                                                                                                                                     
g(L_{0},T)\simeq \ln L_{0}!+\ln \left( 1-e^{-R(T-L_{0})}\right) \nearrow \ln                                                                         
L_{0}!  \label{decay1}                                                                                                                               
\end{equation}%                                                                                                                                      
as $T$ increases and it remains constant once all $L_{0}$-patterns are                                                                               
visible. Therefore, contrarily to the deterministic case (\ref%                                                                                      
{gamma(L,N)_det}), the limit of $g(L_{0},T)$ as $T$ grows is the same for                                                                            
all FPF processes, namely, $\ln L_{0}!$.                                                                                                             
                                                                                                                                                     
\begin{figure*}[tbp]                                                                                                                                 
\begin{center}                                                                                                                                       
\includegraphics[scale=0.45]{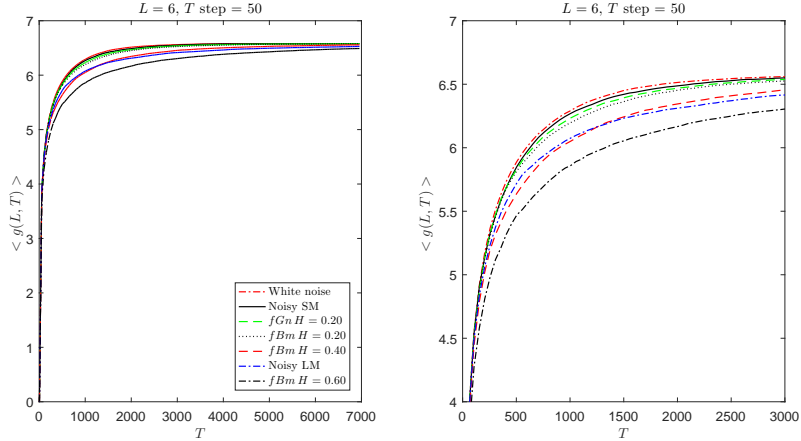}                                                                                                                
\end{center}                                                                                                                                         
\caption{ The averages $\langle g(L,T)\rangle $ of the permutation                                                                                   
complexity function $g(L,T)$ over 35 realizations of the random processes                                                                            
listed in the inset (the same as in Figure~\protect\ref{fig:Figura1}) are                                                                            
plotted vs $T$ (the time series length) for $L=6$. The right panel ($6\leq                                                                           
T\leq 3,000$) is a zoom of the left panel ($6\leq T\leq 7,000$). }                                                                                   
\label{fig:Figura2}                                                                                                                                  
\end{figure*}                                                                                                                                        
                                                                                                                                                     
\medskip                                                                                                                                             
                                                                                                                                                     
A direct application of the \textquotedblleft finite                                                                                                 
length\textquotedblright\ PC function $g(L,T)$ is to discriminate different                                                                          
FPF processes by the different convergence rates of $g(L,T)$ to $\ln L!$ as $%                                                                       
T$ increases. Numerical evidence is shown in Figure~\ref{fig:Figura2} for $%                                                                         
L=6$ and $6\leq T\leq 7000$ in the left panel; the right panel is a zoom of                                                                          
the left panel with $6\leq T\leq 3000$. Here we used the same 35 time                                                                                
series, random processes (Fac1)-(Fac7) and numerical results as in Figure~%                                                                          
\ref{fig:Figura1}, and plotted $\left\langle g(6,T)\right\rangle                                                                                     
=\left\langle \ln \mathcal{A}_{6,T}(\mathbf{X})\right\rangle $ every $\Delta                                                                         
T=50$ points, where $\left\langle \cdot \right\rangle $ denotes again the                                                                            
average over the 35 samples. For all those processes (and any other FPF                                                                              
process for that matter), $\left\langle g(6,T)\right\rangle $ converges to $%                                                                        
\ln 6!\simeq \allowbreak 6$.$579\,3$ as $T$ grows.                                                                                                   
                                                                                                                                                     
The decay exponents $R$ in the approximation (\ref{decay1}) are listed in                                                                            
Table~\ref{tab:Table1} for $L=4,5,6$, $T=7000$ and the random processes                                                                              
(Fac1)-(Fac7) shown in Figure~\ref{fig:Figura2}. Due to the different                                                                                
correlation lengths, the curves $g(L_{0},T)$ may intersect as a result of                                                                            
the different decay rates of the missing patterns.                                                                                                   
                                                                                                                                                     
\begin{table}[t]                                                                                                                                     
\begin{center}                                                                                                                                       
\begin{tabular}{|c|c|c|c|}                                                                                                                           
\hline                                                                                                                                               
Serie & $L = 4$ & $L = 5$ & $L = 6$ \\ \hline                                                                                                        
White Noise & $4.43 \times 10^{-2}$ & $8.47 \times 10^{-3}$ & $1.40 \times                                                                           
10^{-3}$ \\                                                                                                                                          
$fGn H = 0.20$ & $4.93 \times 10^{-2}$ & $8.17 \times 10^{-3}$ & $1.21                                                                               
\times 10^{-3}$ \\                                                                                                                                   
$fBm H = 0.20$ & $4.20 \times 10^{-2}$ & $7.52 \times 10^{-3}$ & $1.13                                                                               
\times 10^{-3}$ \\                                                                                                                                   
$fBm H = 0.40$ & $3.76 \times 10^{-2}$ & $6.32 \times 10^{-3}$ & $8.12                                                                               
\times 10^{-4}$ \\                                                                                                                                   
$fBm H = 0.60$ & $3.24 \times 10^{-2}$ & $4.07 \times 10^{-3}$ & $5.05                                                                               
\times 10^{-4}$ \\                                                                                                                                   
Noisy LM & $3.36 \times 10^{-2}$ & $5.55 \times 10^{-3}$ & $7.54 \times                                                                              
10^{-4}$ \\                                                                                                                                          
Noisy SM & $4.43 \times 10^{-2}$ & $8.09 \times 10^{-3}$ & $1.30 \times                                                                              
10^{-3}$ \\ \hline                                                                                                                                   
\end{tabular}%                                                                                                                                       
\end{center}                                                                                                                                         
\caption{The decay exponent $R$ in Equations (\protect\ref{decay0})-(\protect                                                                        
\ref{decay1}) for $L=4,5,6$, $T=7,000$, and the random processes plotted in                                                                          
Figure~\protect\ref{fig:Figura2}.}                                                                                                                   
\label{tab:Table1}                                                                                                                                   
\end{table}                                                                                                                                          
                                                                                                                                                     
Let us finally mention that Equation (\ref{decay rate}) has recently been                                                                            
generalized to \cite{Olivares2019A,Olivares2019B}                                                                                                    
\begin{equation}                                                                                                                                     
\mathcal{M}_{L,N}(\mathbf{X})=C\exp \left( -RN^{\beta }\right) .                                                                                     
\label{decay rate stretch}                                                                                                                           
\end{equation}%                                                                                                                                      
The \textit{stretching exponent} $0<\beta \leq 1$ depends on $L$ as well as                                                                          
on the underlying random process; for $\beta =1$ one recovers Equation (\ref{decay                                                                            
rate}).                                                                                                                                              
                                                                                                                                                     
\bigskip                                                                                                                                             
                                                                                                                                                     
\subsection{The sub-factorial class}                                                                                                                 
                                                                                                                                                     
\label{sec53}                                                                                                                                        
                                                                                                                                                     
To wrap up this numerical section, we consider the sub-factorial class as                                                                            
well. To this end, we revisit $\mathbf{X}_{\mathfrak{p}}=(X_{t})_{t\geq 0}$,                                                                         
the \textquotedblleft not-so-noisy\ measurement of a periodic                                                                                        
signal\textquotedblright\ of period $\mathfrak{p}\geq 2$ introduced in                                                                               
Example \ref{ExampleSubfact}, since this is the only instance of a                                                                                   
sub-factorial process that we know of. As compared to the factorial                                                                                  
processes studied in Sections \ref{sec51} and \ref{sec52}, $\mathbf{X}_{%                                                                            
\mathfrak{p}}$ has two peculiarities: (i) its conceptual simplicity makes                                                                            
most details amenable to analytical scrutiny, and (ii) it is a                                                                                       
cyclostationary process, i.e., its statistical properties vary periodically                                                                          
with time, so that it can be viewed as a random process composed of $%                                                                               
\mathfrak{p}$ interleaved stationary processes $(X_{\nu \mathfrak{p}+\mu                                                                             
})_{\nu \geq 0}$, $0\leq \mu \leq \mathfrak{p}-1$. Likewise, it turns out                                                                            
that the ordinal representations of $\mathbf{X}_{\mathfrak{p}}$\ decompose                                                                           
into $\mathfrak{p}$ sequences of representations by means of $(\nu \mathfrak{%                                                                       
p}+\mu )$-patterns, $\nu \in \mathbb{N},$ such that $\mathbf{X}_{\mathfrak{p}%                                                                       
}$ has a PC function $g(L)$ for each of those sequences.                                                                                             
                                                                                                                                                     
\begin{figure*}[tbp]                                                                                                                                 
\begin{center}                                                                                                                                       
\includegraphics[scale=0.50]{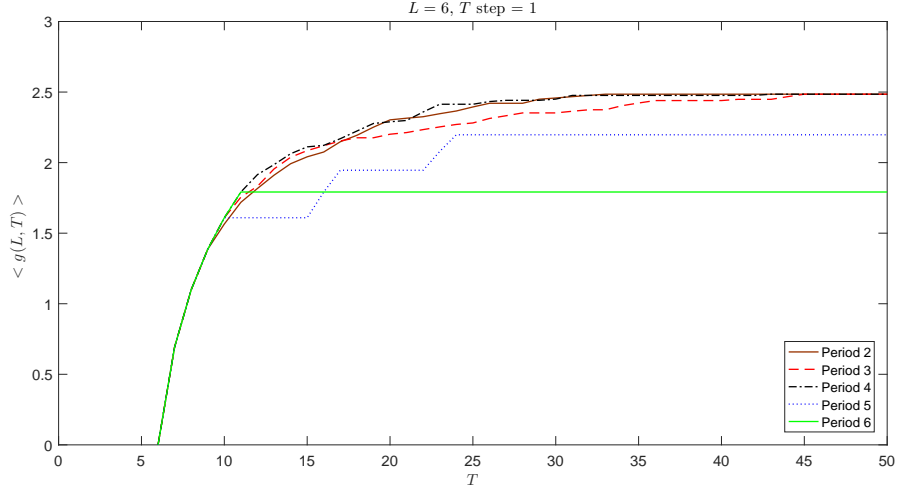}                                                                                                                
\end{center}                                                                                                                                         
\caption{ The averages $\langle g(6,T)\rangle $ of the permutation                                                                                   
complexity function $g(6,T)$ over 35 realizations of the sub-factorial                                                                               
processes $\mathbf{X}_{\mathfrak{p}}$ of periods $2\leq \mathfrak{p}\leq 6$                                                                          
are plotted vs $T$ (the time series length) for $6\leq T\leq 50$. }                                                                                  
\label{fig:Figura3}                                                                                                                                  
\end{figure*}                                                                                                                                        
                                                                                                                                                     
We start our numerical analysis of $\mathbf{X}_{\mathfrak{p}}$ with the                                                                              
finite length PC function $g(L,T)$ for $L=6$. Figure~\ref{fig:Figura3}                                                                               
depicts $g(6,T)$ vs $T$ for $6\leq T\leq 50$ and $2\leq \mathfrak{p}\leq 6$.                                                                         
Here we already recognize that these processes are not factorial because $%                                                                          
g(6,T)<2.5$ in all cases, whereas $g(6,T)\nearrow \ln 6!\simeq \allowbreak                                                                           
6.\,\allowbreak 58$ for factorial processes (see Section \ref{sec52}). Since                                                                         
they are not deterministic either (which can be simply checked by a return                                                                           
map), then they must be sub-factorial. The curves $g(6,T)$ in Figure \ref%                                                                           
{fig:Figura3} are visibly distinct, although the ones corresponding to the                                                                           
periods $\mathfrak{p}=2,3,4$ run close to each other. The reason for this is                                                                         
that $g(6,50)=\ln \mathcal{A}_{6}(\mathbf{X}_{\mathfrak{p}})$, where,                                                                                
according to Table \ref{tab:Table2}, $\ln \mathcal{A}_{6}(\mathbf{X}_{%                                                                              
\mathfrak{p}})=\ln 12\simeq 2.485$ for $\mathfrak{p}=2,3,4$, while $\ln                                                                              
\mathcal{A}_{6}(\mathbf{X}_{5})=\ln 9\simeq 2.197$ and $\ln \mathcal{A}_{6}(%                                                                        
\mathbf{X}_{6})=\ln 6\simeq 1.792$. Also, contrarily to Figure \ref%                                                                                 
{fig:Figura2} for the factorial processes (Fac1)-(Fac7), the curves $g(6,T)$                                                                         
in Figure~\ref{fig:Figura3} do not follow the exponential ansatz (\ref%                                                                              
{decay1}) but rather oscillate before leveling off. This is due to the                                                                               
periodicity of the probability distributions of the $L$-patterns (in                                                                                 
particular, of the $6$-patterns) as the window $x_{t},x_{t+1},...,x_{t+5}$                                                                           
slides from $t=1$ to $t=T-5$.                                                                                                                        
                                                                                                                                                     
According to Equation (\ref{Z_sub}), to compute the generalized permutation                                                                          
entropy $Z_{\text{sub},\alpha }^{\ast }(X_{0}^{L})$ of a sub-factorial                                                                               
process $\mathbf{X}=(X_{t})_{t\geq 0}$ with PC function $g(L)=cL\ln L$, $%                                                                           
0<c<1$, we need to know the constant $c$. With this objective, let $L=\nu                                                                            
\mathfrak{p}+\mu $, where $\nu =\left\lfloor L/\mathfrak{p}\right\rfloor \in                                                                         
\mathbb{N}$ and $\mu =L~\mathrm{{mod}~}\mathfrak{p}\in \{0,1,...,\mathfrak{p}%                                                                       
-1\}$; $\mu =0$ (i.e., $L=\nu \mathfrak{p}$) is the case considered in                                                                               
Example \ref{ExampleSubfact}. An argument similar to the one used there                                                                              
shows that Equation (\ref{formula}) generalizes to%                                                                                                  
\begin{equation}                                                                                                                                     
\mathcal{A}_{L}(\mathbf{X}_{\mathfrak{p}})=N_{1}(\nu ,\mu )+N_{2}(\nu ,\mu )                                                                         
\label{A_L(nu,mu)}                                                                                                                                   
\end{equation}%                                                                                                                                      
where%                                                                                                                                               
\begin{equation}                                                                                                                                     
N_{1}(\nu ,\mu )=(\mathfrak{p}-\mu )[(\nu +1)!]^{\mu }[\nu !]^{\mathfrak{p}%                                                                         
-\mu -1},\;N_{2}(\nu ,\mu )=\mu \lbrack (\nu +1)!]^{\mu -1}[\nu !]^{%                                                                                
\mathfrak{p}-\mu }.  \label{formula2}                                                                                                                
\end{equation}%                                                                                                                                      
Note that $N_{2}(\nu ,0)=0$ for all $\nu $, so that we recover Equation (\ref%                                                                       
{formula}),                                                                                                                                          
\begin{equation}                                                                                                                                     
\mathcal{A}_{\nu \mathfrak{p}}(\mathbf{X}_{\mathfrak{p}})=N_{1}(\nu ,0)=%                                                                            
\mathfrak{p}(\nu !)^{\mathfrak{p}-1},  \label{formula3}                                                                                              
\end{equation}%                                                                                                                                      
from Equation (\ref{A_L(nu,mu)}). Table~\ref{tab:Table2} summarizes $%                                                                               
\mathcal{A}_{L}(\mathbf{X}_{\mathfrak{p}})$ for moderate values of $%                                                                                
\mathfrak{p}$ and $L$.                                                                                                                               
                                                                                                                                                     
\begin{table}[htbp]                                                                                                                                  
\centering                                                                                                                                           
\begin{tabular}{|c|c|c|c|c|r|r|r|r|r|r|r|r|r|}                                                                                                       
\hline                                                                                                                                               
\multicolumn{14}{|c|}{\textbf{ALLOWED PATTERNS}} \\ \hline                                                                                           
\multirow{2}[0]{*}{\textbf{$\mathfrak{p}$}} & \multicolumn{13}{c|}{\textbf{%                                                                         
Window width ($L$)}} \\ \cline{2-14}                                                                                                                 
& \textbf{2} & \textbf{3} & \textbf{4} & \textbf{5} & \multicolumn{1}{c|}{%                                                                          
\textbf{6}} & \multicolumn{1}{c|}{\textbf{7}} & \multicolumn{1}{c|}{\textbf{8%                                                                       
}} & \multicolumn{1}{c|}{\textbf{9}} & \multicolumn{1}{c|}{\textbf{10}} &                                                                            
\multicolumn{1}{c|}{\textbf{11}} & \multicolumn{1}{c|}{\textbf{12}} &                                                                                
\multicolumn{1}{c|}{\textbf{13}} & \multicolumn{1}{c|}{\textbf{14}} \\ \hline                                                                        
\textbf{2} & \multicolumn{1}{r|}{2} & \multicolumn{1}{r|}{3} &                                                                                       
\multicolumn{1}{r|}{4} & \multicolumn{1}{r|}{8} & 12 & 30 & 48 & 144 & 240 &                                                                         
840 & 1,440 & 5,760 & 10,080 \\                                                                                                                      
\textbf{3} & - & \multicolumn{1}{r|}{3} & \multicolumn{1}{r|}{5} &                                                                                   
\multicolumn{1}{r|}{8} & 12 & 28 & 60 & 108 & 324 & 864 & 1,728 & 6,336 &                                                                            
20,160 \\                                                                                                                                            
\textbf{4} & - & - & \multicolumn{1}{r|}{4} & \multicolumn{1}{r|}{7} & 12 &                                                                          
20 & 32 & 80 & 192 & 432 & 864 & 2,808 & 8,640 \\                                                                                                    
\textbf{5} & - & - & - & \multicolumn{1}{r|}{5} & 9 & 16 & 28 & 48 & 80 & 208                                                                        
& 528 & 1,296 & 3,024 \\                                                                                                                             
\textbf{6} & - & - & - & - & 6 & 11 & 20 & 36 & 64 & 112 & 192 & 512 & 1,344                                                                         
\\ \hline                                                                                                                                            
\end{tabular}%                                                                                                                                       
\caption{The number of allowed $L$-patterns $\mathcal{A}_{L}(\mathbf{X}_{                                                                            
\mathfrak{p}})$ for periods $2\leq \mathfrak{p}\leq 6$ and $\mathfrak{p}\leq                                                                         
L\leq 14$.}                                                                                                                                          
\label{tab:Table2}                                                                                                                                   
\end{table}                                                                                                                                          
                                                                                                                                                     
Therefore,                                                                                                                                           
\begin{eqnarray}                                                                                                                                     
\mathcal{A}_{L}(\mathbf{X}_{\mathfrak{p}}) &=&[(\nu +1)!]^{\mu -1}[\nu !]^{%                                                                         
\mathfrak{p}-\mu -1}\left[ (\mathfrak{p}-\mu )(\nu +1)!+\mu v!\right]  \notag                                                                        
\\                                                                                                                                                   
&=&[(\nu +1)!]^{\mu -1}[\nu !]^{\mathfrak{p}-\mu -1}(\mathfrak{p}-\mu )(\nu                                                                          
+1)!\left[ 1+\tfrac{\mu }{(\mathfrak{p}-\mu )(\nu +1)}\right]                                                                                        
\label{formula4}                                                                                                                                     
\end{eqnarray}%                                                                                                                                      
and%                                                                                                                                                 
\begin{equation}                                                                                                                                     
\ln \mathcal{A}_{L}(\mathbf{X}_{\mathfrak{p}})=(\mu -1)\ln (\nu +1)!+(%                                                                              
\mathfrak{p}-\mu -1)\ln \nu !+\ln (\mathfrak{p}-\mu )+\ln (\nu +1)!+\ln %                                                                            
\left[ 1+\tfrac{\mu }{(\mathfrak{p}-\mu )(\nu +1)}\right] .                                                                                          
\label{formula5b}                                                                                                                                    
\end{equation}%                                                                                                                                      
Use now Stirling's formula (\ref{Stirling}) to derive%                                                                                               
\begin{eqnarray}                                                                                                                                     
\ln \mathcal{A}_{L}(\mathbf{X}_{\mathfrak{p}}) &\sim &(\mu -1)(\nu +1)\ln                                                                            
(\nu +1)+(\mathfrak{p}-\mu -1)\nu \ln \nu +(\nu +1)\ln (\nu +1)  \notag \\                                                                           
&=&\mu (\nu +1)\ln (\nu +1)+(\mathfrak{p}-\mu -1)\nu \ln \nu                                                                                         
\label{formula6b}                                                                                                                                    
\end{eqnarray}%                                                                                                                                      
when $\nu \rightarrow \infty $ ($\mu $ is bounded by $\mathfrak{p}-1$).                                                                              
Equation (\ref{formula6b}) leads to the following two cases.                                                                                         
                                                                                                                                                     
\begin{enumerate}                                                                                                                                    
\item If $\mu =0$, i.e., $L=\nu \mathfrak{p}$, then%                                                                                                 
\begin{equation}                                                                                                                                     
\ln \mathcal{A}_{\nu \mathfrak{p}}(\mathbf{X}_{\mathfrak{p}})\sim (\mathfrak{%                                                                       
p}-1)\nu \ln \nu =(\mathfrak{p}-1)\tfrac{L}{\mathfrak{p}}\ln \tfrac{L}{%                                                                             
\mathfrak{p}}\sim \tfrac{\mathfrak{p}-1}{\mathfrak{p}}L\ln L.                                                                                        
\label{formula7}                                                                                                                                     
\end{equation}%                                                                                                                                      
This means that%                                                                                                                                     
\begin{equation}                                                                                                                                     
g(L=\nu \mathfrak{p})=cL\ln L\text{\ with\ }c=\tfrac{\mathfrak{p}-1}{%                                                                               
\mathfrak{p}},  \label{formula8}                                                                                                                     
\end{equation}%                                                                                                                                      
in accordance with Example \ref{ExampleSubfact}. Numerical simulations                                                                               
confirm that, as expected, the $N_{1}(\nu ,\mu )$ allowed $L$-patterns for $%                                                                        
\mathbf{X}_{\mathfrak{p}}$ are equiprobable, hence%                                                                                                  
\begin{equation}                                                                                                                                     
R_{\alpha }(p_{u})=R_{0}(p_{u})=\ln N_{1}(L/\mathfrak{p},0)=(\mathfrak{p}%                                                                           
-1)\ln [\mathfrak{p}(L/\mathfrak{p})!]  \label{formula9}                                                                                             
\end{equation}%                                                                                                                                      
for all $\alpha >0$, see Equation (\ref{formula3}).                                                                                                  
                                                                                                                                                     
\item If $1\leq \mu \leq \mathfrak{p}-1$, i.e., $L=\nu \mathfrak{p}+\mu $                                                                            
with $\mu \neq 0$, then%                                                                                                                             
\begin{eqnarray}                                                                                                                                     
\ln \mathcal{A}_{\nu \mathfrak{p}+\mu }(\mathbf{X}_{\mathfrak{p}}) &\sim                                                                             
&\mu (\nu +1)\ln (\nu +1)=\mu \left( \tfrac{L-\mu }{\mathfrak{p}}+1\right)                                                                           
\ln \left( \tfrac{L-\mu }{\mathfrak{p}}+1\right)  \notag \\                                                                                          
&\sim &\mu \tfrac{L}{\mathfrak{p}}\ln \tfrac{L}{\mathfrak{p}}\sim \tfrac{\mu                                                                         
}{\mathfrak{p}}L\ln L.  \label{formula10}                                                                                                            
\end{eqnarray}%                                                                                                                                      
This means that                                                                                                                                      
\begin{equation}                                                                                                                                     
g(L=\nu \mathfrak{p}+\mu )=cL\ln L\text{\ with\ }c=\tfrac{\mu }{\mathfrak{p}}%                                                                       
\in \left\{ \tfrac{1}{\mathfrak{p}},\tfrac{2}{\mathfrak{p}},...,\tfrac{%                                                                             
\mathfrak{p}-1}{\mathfrak{p}}\right\} \text{.}  \label{formula11}                                                                                    
\end{equation}%                                                                                                                                      
Numerical simulations and theoretical insight show that, in this case, the                                                                           
probability distribution of the allowed $L$-patterns for $\mathbf{X}_{%                                                                              
\mathfrak{p}}$ is composed of $N_{1}(\nu ,\mu )$ $L$-patterns of probability                                                                         
$P_{1}=(\mathfrak{p}-\mu )/\mathfrak{p}N_{1}$, and $N_{2}(\nu ,\mu )$ $L$%                                                                           
-patterns of probability $P_{2}=\mu /\mathfrak{p}N_{2}$. It follows,%                                                                                
\begin{equation}                                                                                                                                     
R_{1}(p)=-N_{1}P_{1}\ln P_{1}-N_{2}P_{2}\ln P_{2},\;\;\text{and\ \ }%                                                                                
R_{\alpha }(p)=\tfrac{1}{1-\alpha }\ln \left( N_{1}P_{1}^{\alpha                                                                                     
}+N_{2}P_{2}^{\alpha }\right)  \label{formula13}                                                                                                     
\end{equation}%                                                                                                                                      
for $\alpha >0$, $\alpha \neq 1$, where $p$ is the probability distribution                                                                          
of the allowed ordinal $L$-patterns for $\mathbf{X}_{\mathfrak{p}}$.                                                                                 
\end{enumerate}                                                                                                                                      
                                                                                                                                                     
From Equations (\ref{formula8}) and (\ref{formula11}) we obtain%                                                                                     
\begin{equation}                                                                                                                                     
c=\tfrac{\mathfrak{p}-1}{\mathfrak{p}}\text{\ \ for\ \ }L=0,\mathfrak{p}-1~%                                                                         
\mathrm{{mod}~}\mathfrak{p}.  \label{formula12}                                                                                                      
\end{equation}%                                                                                                                                      
In particular, $c=1/2$ for $\mathfrak{p}=2$ and $\mu =0,1$, i.e., the                                                                                
process $\mathbf{X}_{2}$ has a unique PC function $g(L)=(1/2)L\ln L$,                                                                                
regardless of whether we use ordinal patterns of even or odd lengths.                                                                                
                                                                                                                                                     
\begin{figure*}[tbp]                                                                                                                                 
\begin{center}                                                                                                                                       
\includegraphics[scale=0.45]{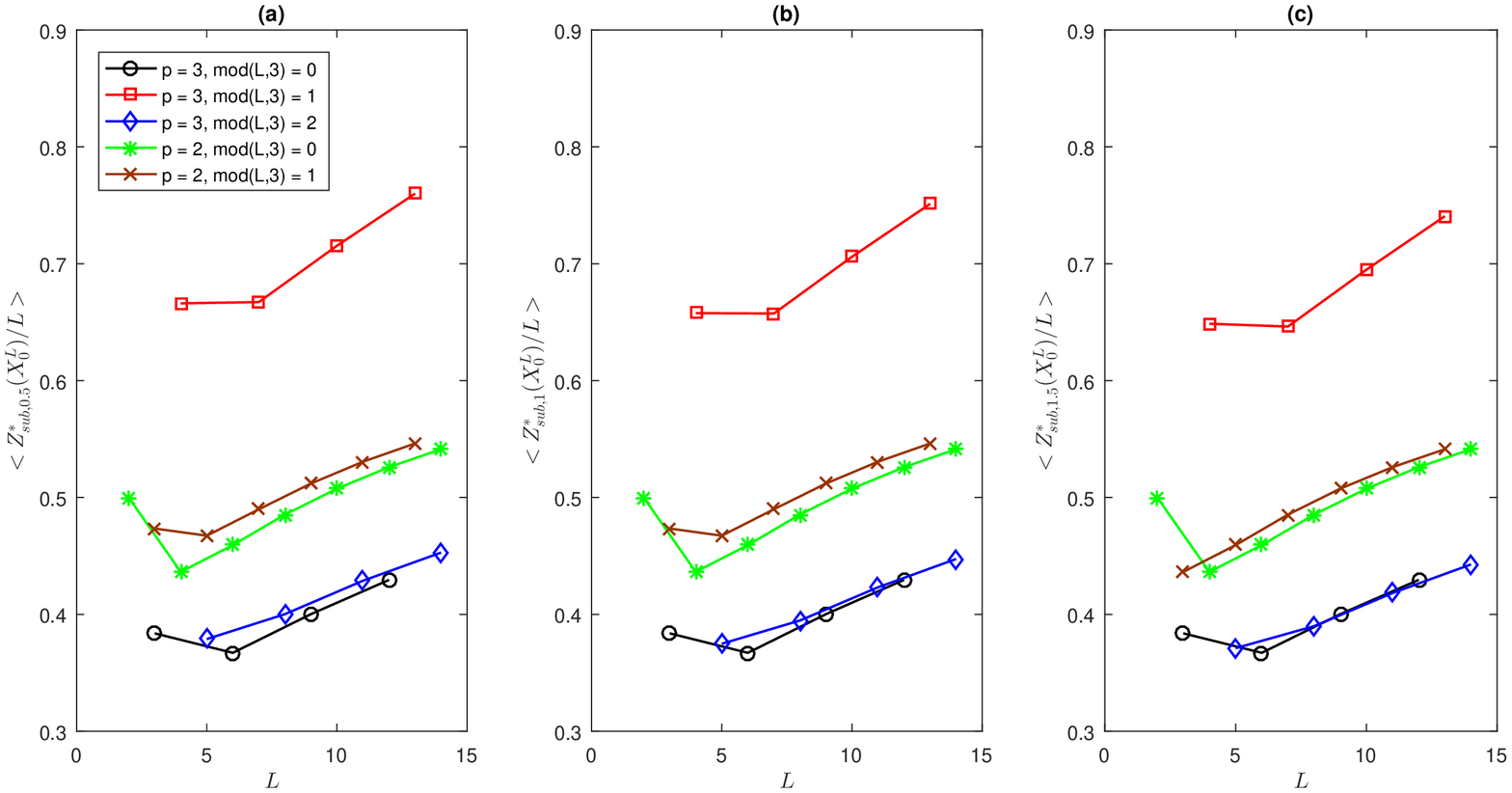}                                                                                                                
\end{center}                                                                                                                                         
\caption{ The averages $\langle Z_{\text{sub},\protect\alpha }^{\ast                                                                                 
}(X_{0}^{L})/L\rangle $ of $Z_{\text{sub},0.5}^{\ast }(X_{0}^{L})/L$ (a), $%                                                                         
Z_{\text{sub},1}^{\ast }(X_{0}^{L})/L$ (b) and $Z_{\text{sub},1}^{\ast                                                                               
}(X_{0}^{L})/L$ (c) over 35 realizations of the sub-factorial processes $%                                                                           
\mathbf{X}_{2,0}$, $\mathbf{X}_{2,1}$, $\mathbf{X}_{3,0}$, $\mathbf{X}_{3,1}$                                                                        
and $\mathbf{X}_{3,2}$ are plotted vs $L$ for $2\leq L\leq 14$. By                                                                                   
definition, each process $\mathbf{X}_{\mathfrak{p},\protect\mu }$, $0\leq                                                                            
\protect\mu \leq \mathfrak{p}-1$, uses ordinal representations $\mathcal{S}%                                                                         
_{L}$ of lengths $L=\protect\nu \mathfrak{p}+\protect\mu $, $\protect\nu \in                                                                         
\mathbb{N}$, to calculate the corresponding complexity function $g(L)=cL\ln                                                                          
L $. Here, $c=1/2$ for $\mathbf{X}_{2,0}$ and $\mathbf{X}_{2,1}$, $c=2/3$                                                                            
for $\mathbf{X}_{3,0}$ and $\mathbf{X}_{3,2}$, and $c=1/3$ for $\mathbf{X}%                                                                          
_{3,1}$. }                                                                                                                                           
\label{fig:Figura4}                                                                                                                                  
\end{figure*}                                                                                                                                        
                                                                                                                                                     
In conclusion, $\mathbf{X}_{\mathfrak{p}}$ generates $\mathfrak{p}$                                                                                  
sub-factorial processes $\mathbf{X}_{\mathfrak{p},\mu }$, $0\leq \mu \leq                                                                            
\mathfrak{p}-1$, according to the ordinal representations used to discretize                                                                         
the realizations of $\mathbf{X}_{\mathfrak{p}}$. Thus, if we use ordinal                                                                             
patterns of lengths $L=\nu \mathfrak{p}$, $\nu \in \mathbb{N}$, the result                                                                           
is a sub-factorial process $\mathbf{X}_{\mathfrak{p},0}$ with%                                                                                       
\begin{equation}                                                                                                                                     
\ln \mathcal{A}_{L=\nu \mathfrak{p}}(\mathbf{X}_{\mathfrak{p},0})\sim \tfrac{%                                                                       
\mathfrak{p}-1}{\mathfrak{p}}L\ln L.  \label{formula14}                                                                                              
\end{equation}%                                                                                                                                      
Otherwise, if we use ordinal patterns of lengths $L=\nu \mathfrak{p}+\mu $,                                                                          
with $\nu \in \mathbb{N}$ and $1\leq \mu \leq \mathfrak{p}-1$ fixed, the                                                                             
result is a sub-factorial process $\mathbf{X}_{\mathfrak{p},\mu }$ with%                                                                             
\begin{equation}                                                                                                                                     
\ln \mathcal{A}_{L=\nu \mathfrak{p}+\mu }(\mathbf{X}_{\mathfrak{p},\mu                                                                               
})\sim \tfrac{\mu }{\mathfrak{p}}L\ln L.  \label{formula15}                                                                                          
\end{equation}%                                                                                                                                      
More formally, we say that the cyclostationary process $\mathbf{X}_{%                                                                                
\mathfrak{p}}$ generates the processes $(\mathbf{X}_{\mathfrak{p}},\{%                                                                               
\mathcal{S}_{\nu \mathfrak{p}+\mu }\}_{\nu \geq 1})=:\mathbf{X}_{\mathfrak{p}%                                                                       
,\mu }$, $0\leq \mu \leq \mathfrak{p}-1$, where $\{\mathcal{S}_{\nu                                                                                  
\mathfrak{p}+\mu }\}_{\nu \geq 1}$ is the subsequence of ordinal                                                                                     
representations used to compute the PC function $g(L)$. In our case, the                                                                             
choices $\{\mathcal{S}_{\nu \mathfrak{p}+\mu }\}_{\nu \geq 1}$ with $\mu $                                                                           
fixed are justified because then the corresponding PC function exists.                                                                               
                                                                                                                                                     
Figure \ref{fig:Figura4}\ shows the average $\langle Z_{\text{sub},\alpha                                                                            
}^{\ast }(X_{0}^{L})/L\rangle $ of the permutation entropy $Z_{\text{sub}%                                                                           
,\alpha }^{\ast }(X_{0}^{L})/L=(e^{\mathcal{L}[R_{\alpha }(p)/c]}-1)/L$ over                                                                         
35 realizations of the sub-factorial processes $\mathbf{X}_{2,\mu }$, $%                                                                             
\mathbf{X}_{3,\mu }$ and $2\leq L\leq 14$, where $\alpha =0.5$ (a), $\alpha                                                                          
=1$ (b) and $\alpha =1.5$ (c). According to Equations (\ref{formula8}) and (%                                                                        
\ref{formula11}), $c=1/2$ for $\mathbf{X}_{2,0}$ and $\mathbf{X}_{2,1}$, $%                                                                          
c=2/3$ for $\mathbf{X}_{3,0}$ and $\mathbf{X}_{3,2}$, while $c=1/3$ for $%                                                                           
\mathbf{X}_{3,1}$. At variance with Figure \ref{fig:Figura1}, the curves in                                                                          
Figure \ref{fig:Figura4} are further apart the lower $\alpha $ is. In                                                                                
particular, we see overlaps of the processes $\mathbf{X}_{3,0}$ and $\mathbf{%                                                                       
X}_{3,2}$ in panels (b) and (c) that, however, are resolved in panel (a).                                                                            
This again illustrates how the parameter $\alpha $ can help when it comes to                                                                         
applications.

\section{Conclusion}

\label{sec6}

Permutation entropy is a popular tool for characterizing time series that
depends on the size of the sliding window used to define the permutations
its name refers to. If the data has been output by a dynamical system, the
conventional permutation entropy rate (\ref{h*_mu}) converges with
increasing window sizes to the Kolmogorov-Sinai entropy of the dynamics
(Theorem \ref{Thm1}(a)). But if the process is noisy deterministic or
random, then that entropy rate diverges in general because the number of
visible ordinal patterns (permutations) can grow super-exponentially, see
Equation (\ref{allowed pat X}). This different growth behavior of the
ordinal patterns and, hence, of the permutation complexity makes possible to
distinguish (noiseless) deterministic processes from noisy processes but
poses a challenge for a unified formulation of permutation complexity and
its measurement across the entire range of processes.

In view of this fact, the objective of this paper was to propose such a
unified formulation. Our approach consisted of two steps. First, we
introduced the permutation complexity (PC) function $g(t)$ in Section \ref%
{sec3}. The asymptotic behavior of $g(t)$ defines the exponential,
sub-factorial and factorial PC classes, the latter being the most
interesting in the application of the ordinal methodology to real-world time
series. A \textquotedblleft finite-length\textquotedblright\ version of the
PC function, $g(L,T)$, where $T$ is the length of a time series and $L\ll T$
is the length of the ordinal patterns, was used in the numerical
simulations, Section \ref{sec52}, to discriminate random processes via its
convergence rate to $\ln L!$ with satisfactory results (see Figure~\ref%
{fig:Figura2}).

Second, we borrowed the concept of $Z$-entropy from statistical mechanics
and complexity theory to generalize permutation entropy from the exponential
PC\ class (the realm of conventional permutation entropy) to the
sub-factorial and factorial PC classes. For this reason, the generalized
permutation entropies go by the name $Z_{g,\alpha }^{\ast }$ (Equation (\ref%
{Zg_entropy})), where the PC function $g(t)$ defines the corresponding
class. $Z$-entropies are group entropies \cite{TJ2020,PT2020} that are
designed to be extensive on the various complexity classes, including
classes not considered here such as the sub-exponential and the
super-factorial. Precisely, the extensivity of the $Z$-entropy entails in
our context that the rate of the generalized permutation entropy, $%
z_{g,\alpha }^{\ast }$ (Equation (\ref{Z*})), converges for the processes in
the sub-factorial class and, foremost, in the factorial class. The
discriminatory power of $Z_{g,\alpha }^{\ast }$ was numerically tested in
Section \ref{sec51} with seven processes belonging to the factorial class,
i.e., $g(t)=t\ln t$. The results, shown in Figure~\ref{fig:Figura1}, were
satisfactory as well. For completion we also analyzed in Section \ref{sec53}
the particularities of a toy model for sub-factorial processes introduced in
Example \ref{ExampleSubfact}. This model is of limited practical interest
but it has the virtue of revealing some subtleties such as the role of
cyclostationarity.

In conclusion,\ we have presented in this paper an integrating approach to
the study and characterization of real-valued processes, whether
deterministic or random, in the ordinal representation. Regarding the
methodology of this approach, the processes are sorted into the exponential,
sub-factorial and factorial complexity classes. Regarding the tools, the
permutation complexity of the processes in each class is measured by the
corresponding permutation entropy, which is the $Z$-entropy of that class.
The result is a \textquotedblleft class-wise\textquotedblright\
generalization of conventional permutation entropy, one per class, whose
rate also converges in the sub-factorial and factorial classes. These
entropic measures of permutation complexity have both fine theoretical
properties and potential in practical applications, in addition to closing
the conceptual gap between deterministic and noisy signals in the ordinal
analysis of time series.

\bigskip

\section{Acknowledgements}

J.M.A. and R.D. were
financially supported by Agencia Estatal de Investigaci\'{o}n, Spain, grant
PID2019-108654GB-I00. J.M.A. was also supported by Generalitat Valenciana,
Spain, grant PROMETEO/2021/063. The research of P.T. has been supported by
the research project PGC2018-094898-B-I00, Ministerio de Ciencia, Innovaci%
\'{o}n y Universidades, Spain, and by the Severo Ochoa Programme for Centres
of Excellence in R\&D (CEX2019-000904-S), Ministerio de Ciencia, Innovaci%
\'{o}n y Universidades, Spain. P.T. is a member of the Gruppo Nazionale di
Fisica Matematica (INDAM), Italy.

\bibliographystyle{unsrt}

\bibliography{Complexity56Arxiv_v3}% Produces the bibliography via BibTeX.

\providecommand{\noopsort}[1]{}\providecommand{\singleletter}[1]{#1}%
\begin{thebibliography}{10}

\bibitem{Morse1940}
Marston Morse and Gustav~A. Hedlund.
\newblock Symbolic dynamics ii. sturmian trajectories.
\newblock {\em American Journal of Mathematics}, 62(1):1--42, 1940.

\bibitem{Lempel1976}
A.~{Lempel} and J.~{Ziv}.
\newblock On the complexity of finite sequences.
\newblock {\em IEEE Transactions on Information Theory}, 22(1):75--81, 1976.

\bibitem{Ziv1978}
J.~Ziv and A.~Lempel.
\newblock Compression of individual sequences via variable-rate coding.
\newblock {\em {IEEE} Transactions on Information Theory}, 24(5):530--536, sep
  1978.

\bibitem{Amigo2015}
José~M. Amigó, Karsten Keller, and Valentina~A. Unakafova.
\newblock On entropy, entropy-like quantities, and applications.
\newblock {\em Discrete and Continuous Dynamical Systems - B},
  20(1531-3492-2015-10-3301):3301, 2015.

\bibitem{Li2008}
Ming Li and Paul Vitányi.
\newblock {\em An Introduction to Kolmogorov Complexity and Its Applications}.
\newblock Springer-Verlag New York, 4 edition, 2019.

\bibitem{Volchan2002}
Sérgio~B. Volchan.
\newblock What is a random sequence?
\newblock {\em The American Mathematical Monthly}, 109(1):46--63, 2002.

\bibitem{Shen2017}
A.~Shen, V.A. Uspensky, and N.~Vereshchagin.
\newblock {\em Kolmogorov Complexity and Algorithmic Randomness}, volume 220 of
  {\em Mathematical Surveys and Monographs}.
\newblock American Mathematical Society, Rhode Island, 2017.

\bibitem{Downey2019}
Rod Downey and Denis~R. Hirschfeldt.
\newblock Computability and randomness.
\newblock {\em Notices of the American Mathematical Society}, 66:1001--1012,
  Aug 2019.

\bibitem{Amigo2010B}
J.~M. Amigó.
\newblock {\em Permutation Complexity in Dynamical Systems}.
\newblock Springer Series in Synergetics. Springer-Verlag Berlin Heidelberg,
  first edition, 2010.

\bibitem{Amigo2010}
J.~M. Amigó, S.~Zambrano, and M.~A.~F. Sanjuán.
\newblock Permutation complexity of spatiotemporal dynamics.
\newblock {\em {EPL} (Europhysics Letters)}, 90(1):10007, apr 2010.

\bibitem{Monetti2013}
R.~Monetti, J.M. Amigó, T.~Aschenbrenner, and W.~Bunk.
\newblock Permutation complexity of interacting dynamical systems.
\newblock {\em Eur. Phys. J. Spec. Top.}, 222:421 -- 436, 2013.

\bibitem{Bandt2002}
C.~Bandt and B.~Pompe.
\newblock Permutation entropy: A natural complexity measure for time series.
\newblock {\em Phys. Rev. Lett.}, 88:174102, Apr 2002.

\bibitem{Rosso2007}
O.~A. Rosso, H.~A. Larrondo, M.~T. Martin, A.~Plastino, and M.~A. Fuentes.
\newblock Distinguishing noise from chaos.
\newblock {\em Phys. Rev. Lett.}, 99:154102, Oct 2007.

\bibitem{Zunino2012}
L.~Zunino, M.~C. Soriano, and O.~A. Rosso.
\newblock Distinguishing chaotic and stochastic dynamics from time series by
  using a multiscale symbolic approach.
\newblock {\em Phys. Rev. E}, 86:046210, Oct 2012.

\bibitem{Amigo2007}
J.~M. Amigó, S.~Zambrano, and M.~A.~F. Sanjuán.
\newblock True and false forbidden patterns in deterministic and random
  dynamics.
\newblock {\em Europhysics Letters ({EPL})}, 79(5):50001, jul 2007.

\bibitem{Carpi2010}
L.~C. Carpi, P.~M. Saco, and O.~A. Rosso.
\newblock Missing ordinal patterns in correlated noises.
\newblock {\em Physica A: Statistical Mechanics and its Applications},
  389(10):2020 -- 2029, 2010.

\bibitem{Pessa2019}
Arthur A.~B. Pessa and Haroldo~V. Ribeiro.
\newblock Characterizing stochastic time series with ordinal networks.
\newblock {\em Phys. Rev. E}, 100:042304, Oct 2019.

\bibitem{TJ2020}
Piergiulio Tempesta and Henrik~Jeldtoft. Jensen.
\newblock Universality classes and information-theoretic measures of complexity
  via group entropies.
\newblock {\em Scientific Reports}, 10(5952):1--11, 2020.

\bibitem{PT2020}
Piergiulio Tempesta.
\newblock Multivariate group entropies, super-exponentially growing complex
  systems, and functional equations.
\newblock {\em Chaos: An Interdisciplinary Journal of Nonlinear Science},
  30(12):123119, 2020.

\bibitem{Amigo2021}
José~M. Amigó, Roberto Dale, and Piergiulio Tempesta.
\newblock A generalized permutation entropy for noisy dynamics and random
  processes.
\newblock {\em Chaos: An Interdisciplinary Journal of Nonlinear Science},
  31(1):013115, 2021.

\bibitem{PT2011PRE}
P.~Tempesta.
\newblock Group entropies, correlation laws, and zeta functions.
\newblock {\em Phys. Rev. E}, 84:021121, Aug 2011.

\bibitem{PT2016PRA}
P.~Tempesta.
\newblock Formal groups and z-entropies.
\newblock {\em Proceedings of the Royal Society A: Mathematical, Physical and
  Engineering Sciences}, 472(2195):20160143, 2016.

\bibitem{JT2018ENT}
H.~J. Jensen and P.~Tempesta.
\newblock Group entropies: From phase space geometry to entropy functionals via
  group theory.
\newblock {\em Entropy}, 20(10), 2018.

\bibitem{JPPT2018JPA}
H.~J. Jensen, R.~H. Pazuki, G.~Pruessner, and P.~Tempesta.
\newblock Statistical mechanics of exploding phase spaces: ontic open systems.
\newblock {\em Journal of Physics A: Mathematical and Theoretical},
  51(37):375002, aug 2018.

\bibitem{RRT2019PRA}
M.~A. Rodríguez, A.~Romaniega, and P.~Tempesta.
\newblock A new class of entropic information measures, formal group theory and
  information geometry.
\newblock {\em Proceedings of the Royal Society A: Mathematical, Physical and
  Engineering Sciences}, 475(2222):20180633, 2019.

\bibitem{Keller2004}
Karsten Keller and Katharina Wittfeld.
\newblock Distances of time series components by means of symbolic dynamics.
\newblock {\em International Journal of Bifurcation and Chaos},
  14(02):693--703, 2004.

\bibitem{Parlitz2012}
U.~Parlitz, S.~Berg, S.~Luther, A.~Schirdewan, J.~Kurths, and N.~Wessel.
\newblock Classifying cardiac biosignals using ordinal pattern statistics and
  symbolic dynamics.
\newblock {\em Computers in Biology and Medicine}, 42(3):319--327, 2012.
\newblock Computing complexity in cardiovascular oscillations.

\bibitem{Graff2013}
G.~Graff, B.~Graff, A.~Kaczkowska, D.~Makowiecz, J.M. Amigó, J.~Piskorski,
  K.~Narkiewicz, and P.~Guzik.
\newblock Ordinal pattern statistics for the assessment of heart rate
  variability.
\newblock {\em Eur. Phys. J. Spec. Top.}, 222:525 -- 534, 2013.

\bibitem{Amigo2015B}
J.~M. Amigó, K.~Keller, and V.~A. Unakafova.
\newblock Ordinal symbolic analysis and its application to biomedical
  recordings.
\newblock {\em Philosophical Transactions of the Royal Society A: Mathematical,
  Physical and Engineering Sciences}, 373(2034):20140091, 2015.

\bibitem{Cao2004}
Yinhe Cao, Wen-wen Tung, J.~B. Gao, V.~A. Protopopescu, and L.~M. Hively.
\newblock Detecting dynamical changes in time series using the permutation
  entropy.
\newblock {\em Phys. Rev. E}, 70:046217, Oct 2004.

\bibitem{Olivares2019A}
Felipe Olivares, Luciano Zunino, and Dario~G. Pérez.
\newblock Revisiting the decay of missing ordinal patterns in long-term
  correlated time series.
\newblock {\em Physica A: Statistical Mechanics and its Applications},
  534:122100, 2019.

\bibitem{Olivares2019B}
Felipe Olivares, Luciano Zunino, Miguel~C. Soriano, and Dar\'{\i}o~G. P\'erez.
\newblock Unraveling the decay of the number of unobserved ordinal patterns in
  noisy chaotic dynamics.
\newblock {\em Phys. Rev. E}, 100:042215, Oct 2019.

\bibitem{Zunino2016}
L.~Zunino and H.~V. Ribeiro.
\newblock Discriminating image textures with the multiscale two-dimensional
  complexity-entropy causality plane.
\newblock {\em Chaos, Solitons and Fractals}, 91:679 -- 688, 2016.

\bibitem{Chagas2021}
E.~T.~C. {Chagas}, A.~C. {Frery}, O.~A. {Rosso}, and H.~S. {Ramos}.
\newblock Analysis and classification of sar textures using information theory.
\newblock {\em IEEE Journal of Selected Topics in Applied Earth Observations
  and Remote Sensing}, 14:663--675, 2021.

\bibitem{Amigo2013}
J.~M. Amigó, K.~Keller, and J.~Kurths.
\newblock Recent progress in symbolic dynamics and permutation complexity.
\newblock {\em Eur. Phys. J. Spec. Top.}, 222:241 -- 247, 2013.

\bibitem{Zanin2012}
M.~Zanin, L.~Zunino, O.~A. Rosso, and D.~Papo.
\newblock Permutation entropy and its main biomedical and econophysics
  applications: A review.
\newblock {\em Entropy}, 14(8):1553--1577, 2012.

\bibitem{Bandt2002B}
C.~Bandt, G.~Keller, and B.~Pompe.
\newblock Entropy of interval maps via permutations.
\newblock {\em Nonlinearity}, 15(5):1595--1602, aug 2002.

\bibitem{Keller2010}
Karsten Keller and Mathieu Sinn.
\newblock Kolmogorov–sinai entropy from the ordinal viewpoint.
\newblock {\em Physica D: Nonlinear Phenomena}, 239(12):997--1000, 2010.

\bibitem{Amigo2012}
J.~M. Amigó.
\newblock The equality of kolmogorov–sinai entropy and metric permutation
  entropy generalized.
\newblock {\em Physica D: Nonlinear Phenomena}, 241(7):789 -- 793, 2012.

\bibitem{Keller2019}
T.~Gutjahr and K.~Keller.
\newblock Equality of kolmogorov-sinai and permutation entropy for
  one-dimensional maps consisting of countably many monotone parts.
\newblock {\em Discrete and Continuous Dynamical Systems - A}, 39(1078):4207,
  2019.

\bibitem{Bandt2007}
C.~Bandt and F.~Shiha.
\newblock Order patterns in time series.
\newblock {\em Journal of Time Series Analysis}, 28(5):646--665, 2007.

\bibitem{Eckmann1985}
J.~P. Eckmann and D.~Ruelle.
\newblock Ergodic theory of chaos and strange attractors.
\newblock {\em Rev. Mod. Phys.}, 57:617--656, Jul 1985.

\bibitem{Mandelbrot1968}
Benoit~B. Mandelbrot and John W.~Van Ness.
\newblock Fractional brownian motions, fractional noises and applications.
\newblock {\em SIAM Review}, 10(4):422--437, 1968.

\bibitem{Walter2000}
P.~Walters.
\newblock {\em An Introduction to Ergodic Theory}, volume~79 of {\em Graduate
  Texts in Mathematics}.
\newblock Springer-Verlag New York, first edition, 1982.

\bibitem{AmigoKeller2013}
J.~M. Amigó and K.~Keller.
\newblock Permutation entropy: One concept, two approaches.
\newblock {\em Eur. Phys. J. Spec. Top.}, 222:263 -- 273, 2013.

\bibitem{Amigo2008B}
J.~M. Amigó and M.~B. Kennel.
\newblock Forbidden ordinal patterns in higher dimensional dynamics.
\newblock {\em Physica D: Nonlinear Phenomena}, 237(22):2893 -- 2899, 2008.

\bibitem{Apostol1974}
T.~Apostol.
\newblock {\em Mathematical Analysis}.
\newblock Addison Wesley Longman, Menlo Park CA, second edition, 1974.

\bibitem{Shannon}
C.~E. Shannon.
\newblock A mathematical theory of communication.
\newblock {\em Bell System Technical Journal}, 27(3):379--423, 1948.

\bibitem{Shannon2}
C.~E. Shannon and W.~Weaver.
\newblock {\em The mathematical Theory of Communication}.
\newblock University of Illinois Press, Urbana, Illinois, first edition, 1949.

\bibitem{Khinchin}
A.I.A. Khinchin.
\newblock {\em Mathematical Foundations of Information Theory}.
\newblock Dover Books on Mathematics. Dover Publications, 1957.

\bibitem{Amigo2018}
José~M. Amigó, Sámuel~G. Balogh, and Sergio Hernández.
\newblock A brief review of generalized entropies.
\newblock {\em Entropy}, 20(11), 2018.

\bibitem{Ilic2021}
V.~M. Ili{\'{c}}, J.~Korbel, S.~Gupta, and A.~M. Scarfone.
\newblock An overview of generalized entropic forms (a).
\newblock 133(5):50005, mar 2021.

\bibitem{Tsa2009}
Constantino Tsallis.
\newblock {\em Introduction to nonextensive statistical mechanics}.
\newblock Springer-Verlag New York, 2009.

\bibitem{Renyi1960}
A.~Rényi.
\newblock On measures of entropy and information.
\newblock In {\em Proceedings of the Fourth Berkeley Symposium on Mathematical
  Statistics and Probability, Volume 1: Contributions to the Theory of
  Statistics}, pages 547--561, Berkeley, Calif., 1960. University of California
  Press.

\bibitem{Olver2010}
F.W.J. Olver, D.W. Lozier, R.F. Boisvert, and C.W. Clark, editors.
\newblock {\em NIST Handbook of Mathematical Functions}.
\newblock Cambridge University Press, Cambridge UK, Jul 2010.

\bibitem{Jaynes1957}
E.~T. Jaynes.
\newblock Information theory and statistical mechanics.
\newblock {\em Phys. Rev.}, 106:620--630, May 1957.

\bibitem{Pessa2021}
Arthur A.~B. Pessa and Haroldo~V. Ribeiro.
\newblock ordpy: A python package for data analysis with permutation entropy
  and ordinal network methods.
\newblock {\em Chaos: An Interdisciplinary Journal of Nonlinear Science},
  31(6):063110, 2021.

\bibitem{Unakafova2013}
Valentina~A. Unakafova and Karsten Keller.
\newblock Efficiently measuring complexity on the basis of real-world data.
\newblock {\em Entropy}, 15(10):4392--4415, 2013.

\bibitem{Kantz2004}
Holger Kantz and Thomas Schreiber.
\newblock {\em Nonlinear Time Series Analysis}.
\newblock Cambridge University Press, 2 edition, 2003.

\bibitem{Ruette2017}
S.~Ruette.
\newblock {\em Chaos on the Interval}.
\newblock University lecture series. American Mathematical Society, Providence,
  Rhode Island, first edition, 2017.

\bibitem{Schuster1988}
Heinz~Georg Schuster.
\newblock {\em Deterministic chaos: an introduction}.
\newblock VCH; Distribution, USA and Canada, VCH Weinheim, Federal Republic of
  Germany: New York, NY, USA, 2nd rev. ed. edition, 1988.

\end{thebibliography}

%\begin{thebibliography}{1}
%
%	\bibitem{kour2014real}
%	George Kour and Raid Saabne.
%	\newblock Real-time segmentation of on-line handwritten arabic script.
%	\newblock In {\em Frontiers in Handwriting Recognition (ICFHR), 2014 14th
%			International Conference on}, pages 417--422. IEEE, 2014.
%
%	\bibitem{kour2014fast}
%	George Kour and Raid Saabne.
%	\newblock Fast classification of handwritten on-line arabic characters.
%	\newblock In {\em Soft Computing and Pattern Recognition (SoCPaR), 2014 6th
%			International Conference of}, pages 312--318. IEEE, 2014.
%
%	\bibitem{hadash2018estimate}
%	Guy Hadash, Einat Kermany, Boaz Carmeli, Ofer Lavi, George Kour, and Alon
%	Jacovi.
%	\newblock Estimate and replace: A novel approach to integrating deep neural
%	networks with existing applications.
%	\newblock {\em arXiv preprint arXiv:1804.09028}, 2018.
%
%\end{thebibliography}

\end{document}